\documentclass[noacm,sigconf,balance=false]{acmart}
\renewcommand\footnotetextcopyrightpermission[1]{}

\usepackage{enumitem}
\usepackage{algorithm,algpseudocode}
\usepackage{tabularx}
\usepackage{amsfonts,amsmath,amsthm,stmaryrd}
\usepackage{tcolorbox}
\usepackage{subcaption}

\newtheorem{Thm}{Theorem}[section]
\newtheorem{Lem}[Thm]{Lemma}
\newtheorem{Prop}[Thm]{Proposition}

\newtheorem{Def}[Thm]{Definition}
\newtheorem{Exm}[Thm]{Example}

\def\R{\mathbb{R}}
\def\N{\mathbb{N}}

\def\Z{\mathbb{Z}}
\def\F{\mathbb{F}}
\def\G{\mathbb{G}}
\def\M{\mathcal{M}}
\newcommand\Mof[1]{\M\left(#1\right)}
\def\D{\mathcal{D}}
\def\Y{\mathcal{Y}}
\usepackage{bbm}
\def\1{\mathbbm{1}}
\renewcommand{\epsilon}{\varepsilon}

\newcommand{\prob}[1]{\Pr\left[\begin{array}{@{}c@{}}#1\end{array}\right]}

\newcommand{\case}[1]{\left\{\begin{array}{@{}ll@{}}#1\end{array}\right.}

\usepackage[
	operators,
	advantage,
	sets,
	adversary,
	landau,
	probability,
	notions,	
	logic,
	ff,
	mm,
	primitives,
	events,
	complexity,
	oracles,
	asymptotics,
	keys]{cryptocode}
\usepackage{adjustbox}
\microtypesetup{expansion=false}
\newcommand{\round}[1]{\left\lfloor#1 \right\rceil}
\renewcommand{\vec}[1]{\overrightarrow{#1}}

\renewcommand{\sf}[1]{\ifmmode\mathsf{#1}\else\textsf{#1}\fi}
\newcommand{\adbf}[1]{\ifmmode\mathbf{#1}\else\textbf{#1}\fi}
\newcommand{\bs}[1]{\boldsymbol{#1}}

\newcommand{\range}[1]{\left[#1\right]}
\newcommand{\secsh}[1]{\secretshare{#1}}

\newcommand{\checkeq}{\stackrel{?}{=}}

\def\ser{\ensuremath{\sf{Ser}}}
\def\cli{\ensuremath{\sf{Cli}}}
\newcommand{\secretshare}[1]{\ensuremath{\left\llbracket#1\right\rrbracket}}

\usepackage{multirow}

\usepackage{booktabs}
\setlist{nosep}
\usepackage{pifont}

\newcommand{\autofit}[1]{\maxsizebox{\linewidth}{!}{\begin{varwidth}{10\linewidth}#1\end{varwidth}}}

\newcommand{\VDDPCodeURL}{\url{https://anonymous.4open.science/r/SIGMOD2027-897/}}

\allowdisplaybreaks

\begin{document}

\title{VDDP: Verifiable Distributed Differential Privacy \\under the Client-Server-Verifier Setup}

\author{Haochen Sun}
\affiliation{%
  \institution{University of Waterloo}
  \city{Waterloo}
  \country{Canada}}
\email{haochen.sun@uwaterloo.ca}

\author{Xi He}
\affiliation{%
  \institution{University of Waterloo}
  \city{Waterloo}
  \country{Canada}}
\email{xi.he@uwaterloo.ca}

\begin{abstract}
    
Although differential privacy (DP) is widely regarded as the de facto standard for data privacy, its implementation remains vulnerable to unfaithful execution by servers, particularly in distributed settings. In such cases, servers may sample noise from incorrect distributions or generate correlated noise while appearing to follow established protocols. This work addresses these malicious behaviours in a distributed client-server-verifier setup under Verifiable Distributed Differential Privacy (VDDP), a novel framework for the verifiable execution of distributed DP mechanisms. We systematically capture end-to-end security and privacy guarantees against potentially colluding adversarial behaviours of clients, servers, and verifiers by characterizing the connections and distinctions between VDDP and zero-knowledge proofs (ZKPs).

As our primary instantiation of VDDP, the Verifiable Distributed Discrete Laplace Mechanism (VDDLM) achieves a proving-time speedup ranging from 50x to 4,000,000x over existing solutions, with larger gains in stronger-privacy regimes. Furthermore, the VDDLM sampling circuit achieves asymptotically lower overhead than prior constructions and is applicable to a broad range of cryptography-based DP applications.

\end{abstract}

\maketitle
\pagestyle{plain}

\section{Introduction}

In today's data-driven world, vast amounts of information are collected and analyzed to drive decision-making across sectors including healthcare, finance, and marketing. However, privacy and efficiency concerns often make centralized collection and analysis infeasible. Therefore, \emph{distributed data analysis}~\cite{privex, prio, secagg} has emerged as a viable solution. In this approach, data are distributed and processed across multiple nodes, enhancing privacy by reducing the exposure of sensitive information and improving efficiency by leveraging the computational resources of several nodes.

With differential privacy (DP)~\cite{DP,DR14} established as the de facto standard for data privacy, distributed data analysis has also benefited from an additional privacy layer, giving rise to \emph{distributed differential privacy}~\cite{DBLP:conf/eurocrypt/CheuSUZZ19, DBLP:journals/tifs/WeiJWHDLCPW24, DBLP:conf/podc/Dwork21, DBLP:journals/popets/MeisingsethR25}. This approach includes implementations such as multi-party computation (MPC) of differential privacy schemes~\cite{DBLP:conf/uss/BohlerK20, DBLP:conf/ccs/BohlerK21,DBLP:journals/corr/abs-2109-10074,DBLP:conf/ccs/FuW24,DBLP:journals/tdsc/GoryczkaX17,DBLP:conf/nips/HeikkilaLKSTH17} and differentially private federated learning (FL) schemes~\cite{DBLP:journals/corr/abs-2405-08299, DBLP:conf/ccs/Liu00L024, DBLP:conf/ccs/XuZH24, DBLP:conf/ccs/PeinemannKSCM24}. These methods keep individual data points private while allowing accurate and efficient analysis across distributed systems.

While benefiting from the privacy guarantees provided by DP, the distributed nature of the mechanisms attracts malicious participants. Specifically, malicious servers can deviate from the DP mechanism during execution, potentially in a collusive manner, which can further compromise privacy guarantees. This type of deviation is elusive since attackers can attribute any anomalous results to randomness. Under the assumption that a certain portion of servers is honest, maliciously secure MPC protocols have been designed to resist such attacks~\cite{DBLP:conf/pst/AnandanC15, DBLP:conf/eurocrypt/DworkKMMN06, DBLP:conf/ccs/ChampionSU19, DBLP:conf/ccs/WeiYFCW23}. However, to convince an external party (e.g., a data analyst) of the authenticity of the result, this assumption must be lifted, and the honesty of the servers still needs to be verified. This type of verification is particularly difficult because the privacy requirements forbid a thorough examination of the execution.

Most concerningly, in 2026, Sun and He discovered that undetectable backdoors can be covertly inserted into DP libraries, such that the privacy cost of the compromised mechanism can be arbitrarily amplified while remaining computationally indistinguishable to all external parties, including the data analyst~\cite{GPM}. Under the central DP setting, implanting such a backdoor requires the malicious actor to compromise the software supply chain. However, in the distributed setup, malicious servers controlled by an adversary are naturally incentivized to covertly execute compromised mechanisms, which inherently serves the goal of breaking the system's security and privacy guarantees. Verifiable computations, especially zero-knowledge proofs (ZKPs)~\cite{DBLP:conf/crypto/BellareG92, DBLP:conf/crypto/FiatS86, DBLP:conf/crypto/GoldreichMW86, DBLP:conf/stoc/GoldwasserMR85, DBLP:journals/ftsec/Thaler22}, have therefore been identified as a countermeasure to ensure faithful mechanism execution while preserving privacy guarantees.

Verifiable computations have been achieved in prior studies for several DP mechanisms, including the continuous (floating-point) Gaussian mechanism~\cite{DBLP:conf/iclr/ShamsabadiTCBHP24}. Moreover, the Verifiable Distributed Binomial Mechanism (VDBM)~\cite{BC23} for multiple clients and servers has also been developed. These studies use ZKPs, from which verifiability and privacy protections are inherited. However, there remains a need for verifiable DP mechanisms with lower overhead, thereby enhancing scalability and practicality while achieving better privacy--utility trade-offs. For example, the total execution time of the binomial mechanism exceeds 30 minutes to achieve $\left(10^{-3}, 10^{-10}\right)$-DP on a one-dimensional input, while it remains unclear whether the application of discrete cryptographic primitives may exacerbate the numerical issues of continuous Gaussian mechanisms~\cite{DBLP:conf/sp/JinMRO22,DBLP:conf/ccs/Mironov12}.

Furthermore, there is a lack of a formal framework for verifiable differential privacy, especially in the distributed setting with multiple clients and servers. For example, a hospital may have data for multiple patients stored on various servers across different departments. To answer a query from an external analyst, multiple servers jointly compute a DP mechanism on the patients' data. Similar scenarios have been considered in (D)Prio~\cite{prio,dprio} and ShrinkWrap~\cite{shrinkwrap}. However, it remains challenging for the analyst to verify the integrity of the query result, which involves both the validity of each patient's records and the correctness of each server's computation. ZKP protocols, typically involving a single prover and a single verifier, do not capture the interactions among multiple clients holding the data and servers executing the computations. For example, multiple servers may correlate their randomness to bias the output distribution, even if their marginal distributions remain correct. Therefore, verifiable computation of the noise-sampling processes must be resilient against collusion among multiple provers (servers). Additionally, although pioneering cryptographic infrastructure developments have extended ZKPs to the distributed setup (e.g., MPCs)~\cite{DBLP:conf/scn/BaumDO14, DBLP:conf/eurosp/KanjalkarZGM21, DBLP:conf/uss/OzdemirB22}, the notion of distributed DP with respect to each client's local database, with the additional requirement of verifiability, remains unclear. To address these challenges, we present \textbf{Verifiable Distributed Differential Privacy (VDDP)}. Our contributions can be summarized as follows:

\begin{itemize}[leftmargin=*]
    \item We rigorously construct a framework for verifiable computations of DP mechanisms under the distributed client-server-verifier setup, capturing the potential for collusive deviations from the protocol by clients and servers, as well as all privacy leakage from one client to other colluding clients, servers, and the verifier. We explore the relationship between ZKP and DP under our framework, showing that while ZKPs are sufficient, they are not necessary for achieving end-to-end DP. \textbf{(Section~\ref{sec:def})}
    \item We develop the \emph{Verifiable Distributed Discrete Laplace Mechanism (VDDLM)}, providing a novel sampling-circuit construction and a tight privacy analysis for the Discrete Laplace Mechanism. Its asymptotically lower circuit complexity also makes it applicable to future secure DP computation using cryptography, e.g., MPCs and fully homomorphic encryption (FHE). \textbf{(Section~\ref{sec:vddlm})}
    \item We conduct systematic experiments on VDDLM, demonstrating substantial empirical efficiency gains, including a proving-time speedup from $50\times$ to approximately $4{,}000{,}000\times$ over VDBM under the same client-server-verifier setup. \textbf{(Section~\ref{sec:experiments})}
\end{itemize}

\subsection{Related Work} \label{sec:rw}

We compare this study's security and privacy models with the aforementioned studies in Table~\ref{tab:comparison}. Pioneering steps in verifiable executions of differentially private mechanisms involve cryptographic proofs of the correctness of deterministic fundamental computation steps in differentially private database systems such as \textbf{VFuzz}~\cite{DBLP:conf/eurosys/NarayanFPH15} and \textbf{DPrio}~\cite{dprio}. More recent advancements have shifted their focus to correct sampling from noise distributions. For example, verifiable executions of DP-SGD (\textbf{STC+24},~\cite{DBLP:conf/iclr/ShamsabadiTCBHP24}) with floating-point Gaussian mechanisms have been developed. However, the additional privacy cost incurred by floating-point arithmetic and further adaptations to cryptographic backends remain to be addressed. In concurrent work, Franzese et al. (\textbf{FFG+25},~\cite{DBLP:conf/ccs/FranzeseF00JPD25}) also proposed a sampling-circuit construction for generic DP mechanisms based on DP noise that is suitable for MPC and ZKP. However, its dependence on lookup tables necessitates listing all possible noise values. The resulting support size is not $\poly[\log\frac{1}{\epsilon}, \log\frac{1}{\delta}]$ and is therefore arguably superpolynomial, which seriously compromises protocol scalability.

\begin{table}[!htbp]
    \caption{Comparison of security and privacy models with previous work on MPCs of DP mechanisms (\textbf{MPC-DP},~\cite{DBLP:conf/ccs/ChampionSU19,DBLP:conf/uss/BohlerK20,DBLP:conf/ccs/BohlerK21,DBLP:conf/eurocrypt/DworkKMMN06,DBLP:conf/ccs/WeiYFCW23}) and verifiable executions of DP mechanisms~\cite{DBLP:conf/eurosys/NarayanFPH15,dprio,DBLP:conf/iclr/ShamsabadiTCBHP24,BC23}. \emph{VD}, \emph{VC}, \emph{VR}: authenticity of data, correctness of deterministic computation, or correctness of sampling from the prescribed random distributions is verifiable (to an external data analyst); \emph{P}: achieving a $\poly[\log\frac{1}{\epsilon}, \log\frac{1}{\delta}]$ overhead; \emph{N}: resilience against numerical issues of DP due to compatibility with discrete cryptographic primitives; \emph{CSV}: client-server-verifier model.}

    \centering
    \autofit{\begin{tabular}{@{}lccccccc@{}}
    \toprule
        ~ & VD & VC & VR & P & N & CSV \\
        \midrule
        MPC-DP & \ding{56} & \ding{56} & \ding{56} & \ding{52} & \ding{52} & \ding{56} \\
        VFuzz~\cite{DBLP:conf/eurosys/NarayanFPH15} & \ding{52} & \ding{52} & \ding{56} & \ding{52} & \ding{56} & \ding{56} \\
        DPrio~\cite{dprio} & \ding{52} & \ding{56} & \ding{52} & \ding{52} & \ding{52} & \ding{52} \\
        STC+24~\cite{DBLP:conf/iclr/ShamsabadiTCBHP24} & \ding{52} & \ding{52} & \ding{52} & \ding{52} & \ding{56} & \ding{56} \\
        FFG+25~\cite{DBLP:conf/ccs/FranzeseF00JPD25} & \ding{52} & \ding{52} & \ding{52} & \ding{56} & \ding{52} & \ding{56} \\
        VDBM~\cite{BC23} & \ding{52} & \ding{52} & \ding{52} & \ding{56} & \ding{52} & \ding{52} \\
        \textbf{Ours} & \ding{52} & \ding{52} & \ding{52} & \ding{52} & \ding{52} & \ding{52} \\
        \bottomrule
    \end{tabular}}
    \label{tab:comparison}
\end{table}

Most relevant to our work, in 2023, Biswas and Cormode proposed verifiable executions of binomial mechanisms (\textbf{VDBM},~\cite{BC23}) \textbf{under the same setup and threat model} with multiple potentially malicious and colluding clients and servers, and a semi-honest verifier. Therefore, in this study, we focus on formalizing the end-to-end security and privacy guarantees under the framework and improving the efficiency--privacy--utility trade-off by achieving verifiable executions of state-of-the-art discrete DP mechanisms with more efficient protocols.

More broadly, multi-party computation (MPC) for randomness generation~\cite{DBLP:conf/pkc/AmbainisJL04,DBLP:conf/sp/BonehBCGI21} and differential privacy~\cite{DBLP:conf/sigmod/ChowdhuryW0MJ20,DBLP:conf/ccs/BellBGL020,DBLP:conf/eurocrypt/DworkKMMN06,DBLP:conf/ccs/ChampionSU19,DBLP:conf/uss/BohlerK20,DBLP:conf/ccs/BohlerK21,DBLP:journals/corr/abs-2109-10074,DBLP:conf/ccs/WeiYFCW23,DBLP:conf/ccs/FuW24} have laid the foundation for investigating the integrity of DP mechanism execution, especially in distributed settings. Due to concerns related to compatibility with cryptographic primitives, discrete DP mechanisms~\cite{DBLP:conf/innovations/BalcerV18, DBLP:journals/siamcomp/GhoshRS12, dgauss} and distributed DP mechanisms~\cite{cpsgd, Skellam, ddgm, pbm} constructed over them via SecAgg~\cite{secagg} have been considered for constructing secure computations, including MPCs and verifiable executions of DP mechanisms. However, despite similarities in multiple aspects, these approaches are based on the additional assumption of a minimal number of honest servers and do not cover the scenario in which an external data analyst needs to verify the authenticity of the data and the correctness of the computation, especially the randomness involved.

\section{Preliminaries} \label{sec:prelim}

In this study, we assume that all parties involved, including clients, servers, verifiers, and generic adversaries, are probabilistic polynomial-time (PPT) algorithms. We also assume that the number of clients and servers and the size of each database are polynomial in the security parameter $\kappa$. The notation is summarized in Table~\ref{tab:notation}.

\begin{table}[!t]
    \centering
    \caption{Notation}
    \footnotesize 
    \begin{tabularx}{\linewidth}{@{}lX@{}}
    \toprule
        Notation & Definition \\
        \midrule
        $D\simeq D'$ & $D$ and $D'$ are neighbouring databases\\
        $\mathcal{D}, \mathcal{Y}$ & input/output spaces of databases and mechanisms\\
        $\mathcal{Q}$ & the clean (non-private) query function applied to the database\\
        $\mathcal{M}$ & the DP mechanism applied to the database\\
        $\epsilon, \delta$ & privacy parameters of DP mechanisms\\
        $\Delta$ & sensitivity of query $\mathcal{Q}$\\
        $d$ & the output dimension of $\mathcal{Q}$ and $\mathcal{M}$\\
        $\F, \G$ & finite field and cyclic group used by cryptographic primitives\\
        $\left[n\right]$ & set of integers $\left\{0, 1, 2, \dots, n-1\right\}$\\
        $y\sample S$ & $y$ is independently and uniformly sampled from a set $S$\\
        $y\sample \mathcal{P}$ & $y$ is independently sampled from a distribution $\mathcal{P}$\\
        $y\sample \mathcal{F}(x)$ & $y$ is independently sampled from the output distribution of a randomized function $\mathcal{F}(x)$\\
        $Y\sim S/\mathcal{P}/\mathcal{F}(x)$ & a random variable $Y$ follows the distribution as defined above\\
        $\mathcal{P}_1 \equiv \mathcal{P}_2$ & two equivalent probability distributions $\mathcal{P}_1$ and $\mathcal{P}_2$\\
        $\kappa$ & the generic security parameter\\
        $\negl[\kappa]$ & the class of functions that are asymptotically smaller than any inverse polynomial\\
        $n_\cli$, $n_\ser$ & the total number of clients/servers\\
        $\vec{\cli}$, $\vec{\ser}$ & the set of potentially malicious clients/servers\\
        $\cli_j$, $D_j$ & the $j$th client with a sensitive local database $D_j$\\ 
        $\ser_i$ & the $i$th server that computes over the sensitive data\\
        $\secretshare{D_j}_i$ & the secret share of $D_j$ transmitted to $\ser_i$\\
        $I_j$ & the indices of servers receiving secret shares from $\cli_j$\\
        $J_i$ & the indices of clients secret-sharing their data with $\ser_i$\\
        $J^*, I^*$ & the indices of honest clients and servers\\
        $J^\dagger$ & the set of clients contributing to the final output\\
        $H_j$ & the indices of fully honest servers associated with $\cli_j$\\
        $\tau$ & the maximum number of colluding servers tolerated in privacy guarantees\\
        $\mathfrak{C}$ & the collusion of parties considered in the privacy analysis\\
        $\verifier$ & the semi-honest verifier (e.g., data analyst) that is interested in the computation output and verifies its authenticity\\
        $a \checkeq b$ & the verifier checks whether $a=b$\\
        \bottomrule
    \end{tabularx}%
    
    \label{tab:notation}%
\end{table}

\subsection{Differential Privacy (DP)} \label{sec:prelim-dp}
We first present the standard definition, which focuses on the scenario of \emph{central DP} (Definition~\ref{def:dp}). In this setting, a server collects data from multiple clients, executes the DP mechanism, and transmits the output to a data analyst.

\begin{Def}[Differential Privacy~\cite{DR14}] \label{def:dp}
    A randomized mechanism $\mathcal{M}: \mathcal{D} \to \mathcal{Y}$, where $\mathcal{D}$ is the space of input databases and $\mathcal{Y}$ is the space of outputs, is $(\epsilon, \delta)$-differentially private if, for any pair of neighbouring (differing in one record) databases $D\simeq D'$ and any measurable subset $S \subseteq \mathcal{Y}$,
    \begin{equation}
        \prob{\mathcal{M}\left(D\right) \in S} \leq e^\epsilon \prob{\mathcal{M}\left(D'\right) \in S} + \delta.
    \end{equation}
    $\mathcal{M}$ achieves $\epsilon$-\emph{pure differential privacy} when $\delta = 0$. 
\end{Def}

Additionally, due to the application of cryptographic primitives in this study, we consider \emph{computational differential privacy (CDP)}. We use its standard indistinguishability-based notion, IND-CDP, the analogue of Definition~\ref{def:dp} that holds only against PPT adversaries instead of computationally unbounded ones.

\begin{Def}[IND-CDP~\cite{CDP}]\label{def:ind-cdp}
    Given $\epsilon: \N \to \R_+$, a family of randomized mechanisms $\M_\kappa: \D_\kappa \to \Y_\kappa$ is $\epsilon$-IND-CDP if, for any family of PPT adversaries $\left\{\adv_\kappa\right\}_{\kappa\in \N}$ and any polynomial $p$, there exists $\delta\left(\kappa\right)\in \negl[\kappa]$ such that, for all pairs of neighbouring databases $D_\kappa\simeq D_\kappa'$ satisfying $\max\left\{\abs{D_\kappa}, \abs{D_\kappa'}\right\}\leq p\left(\kappa\right)$,
    \begin{equation}
        \prob{\adv_\kappa\left(y_\kappa\right) = 1:\\y_\kappa\sample \M_\kappa\left(D_\kappa\right)} \leq e^{\epsilon(\kappa)}\prob{\adv_\kappa\left(y_\kappa\right) = 1:\\y_\kappa\sample \M_\kappa\left(D_\kappa'\right)} + \delta\left(\kappa\right).
    \end{equation}
\end{Def}

\subsubsection{Discrete DP Mechanisms}
To address numerical issues in continuous DP mechanisms (e.g., Laplace and Gaussian) that can cause unintentional privacy leakage~\cite{DBLP:conf/ccs/Mironov12, DBLP:conf/sp/JinMRO22}, we focus on discrete counterparts suitable for cryptographic primitives. The \emph{Discrete Laplace Mechanism (DLM)}~\cite{DBLP:journals/siamcomp/GhoshRS12, DBLP:conf/innovations/BalcerV18} adds noise $\sf{Lap}_\Z(t)$, where
\begin{equation}
\prob{r = \sf{Lap}_\Z(t)} = \frac{e^\frac{1}{t}-1}{e^\frac{1}{t}+1} \cdot\exp\left(-\frac{\abs{r}}{t}\right),\,\forall r \in \Z.
\end{equation} 
For query $\mathcal{Q}: \mathcal{D} \to \Z$ with sensitivity $\Delta$ (i.e., $\abs{\mathcal{Q}(D) - \mathcal{Q}(D')} \leq \Delta$ for neighbouring databases), adding noise from $\sf{Lap}_\Z\left(\frac{\Delta}{\epsilon}\right)$ achieves $\left(\epsilon, 0\right)$-DP. For multi-dimensional queries $\mathcal{Q}: \mathcal{D} \to \Z^d$, adding i.i.d. discrete Laplace noise to each dimension preserves the same guarantee when sensitivity $\Delta$ is defined with the $L_1$ distance~\cite{DR14}. In practical implementations, numerical errors in sampling probabilities must still be explicitly accounted for.

\subsection{Cryptographic Primitives} \label{sec:prelim-crypto}

To facilitate the application of cryptographic primitives, all variables and their arithmetic relations are mapped to a finite field $\F$ of large prime order (e.g., $\abs{\F}$ is a 256-bit integer). We utilize the following cryptographic primitives:

\subsubsection{Pseudo-random Number Generator (PRNG)} \label{sec:prelim-crypto-prng}
A PRNG is considered cryptographically secure if its output is computationally indistinguishable from true randomness~\cite{DBLP:books/crc/KatzLindell2014}, as formalized below.

\begin{Def}\label{def:csprng}
    A PRNG $G: S \to R$, where $S$ is the space of random seeds and $R$ is the space of random outputs, is cryptographically secure if, for any PPT adversary \adv, \begin{equation}
        \abs{\prob{\adv(G(s)) = 1: \\
        s \sample S} - \prob{\adv(r) = 1: \\
        r \sample R}} \in \negl[\kappa].
    \end{equation}
\end{Def}

To verifiably sample many random values in high dimensions, we use PRNGs to streamline the process and reduce communication overhead between the prover and verifier. For efficiency, we use the Legendre pseudo-random function (LPRF)~\cite{DBLP:conf/crypto/Damgard88,DBLP:journals/tosc/BeullensBUV20, DBLP:conf/pqcrypto/BeullensG20, DBLP:conf/csfw/MayZ22} with parallel PRNG construction~\cite{DBLP:conf/sc/SalmonMDS11} to generate fair coins, as applied in previous multi-party computation (MPC) and zero-knowledge proof schemes for machine learning and data management~\cite{DBLP:conf/ccs/0001RRSS16, DBLP:conf/uss/ChangSCKP23}. Given dimensionality $d$, $\sf{LPRF}\left(\cdot; d\right)$ takes a random seed $s \sample \F$ as input and outputs $d$ pseudo-random bits as follows:
\begin{equation}
    \sf{LPRF}\left(s; d\right) := \left(L_s\left(0\right), L_s\left(1\right), \dots, L_s\left(d-1\right)\right),
\end{equation}
where $L_s(k)$ is $1$ if there exists $x$ such that $x^2 = k+s$, and $0$ otherwise. The correctness of $L_s(k)$ can be verified by pre-computing a non-perfect square $t$ of $\F$, and then computing $x := \case{\sqrt{t(k+s)} & \text{if }b = 0 \\ \sqrt{k+s} & \text{if }b = 1}$; thus, it suffices to verify the correctness of the square relation and that $b \in \bin$, i.e.,
\begin{equation}
    x^2 = \left((1-b)t + b\right)(k+s) \land b(1-b) = 0. \label{eq:lprf-proof}
\end{equation}

\subsubsection{Commitments and Zero-Knowledge Proofs} \label{sec:prelim-crypto-zkp}

We utilize zero-knowledge proofs (ZKPs)~\cite{DBLP:books/daglib/0023084, DBLP:journals/iacr/BlumbergTVW14, DBLP:conf/tcc/Lee21, DBLP:journals/iacr/GabizonWC19, DBLP:conf/eurocrypt/ChenBBZ23} to ensure the correctness of the computation without leaking additional information about the sensitive inputs. To protect the inputs without revealing them, we utilize \emph{homomorphic commitment} schemes.

A commitment scheme allows a party to commit to a chosen value while keeping it hidden, with the ability to reveal it later. We denote a commitment to a value $x$ as $c \gets \mathsf{Commit}(x, r)$, where $r$ is a randomly sampled blinding factor.

The schemes we employ (such as Pedersen~\cite{DBLP:conf/crypto/Pedersen91} and KZG~\cite{DBLP:conf/asiacrypt/KateZG10} over a finite field $\F$ and cyclic group $\G$) are \emph{additively homomorphic} on committed values. We denote commitment composition by an abstract operator $\odot$; for example, $\odot$ can be instantiated by multiplication in a multiplicative group $\G$. Given two commitments $c_1 \gets \mathsf{Commit}(x_1, r_1)$ and $c_2 \gets \mathsf{Commit}(x_2, r_2)$, their composition yields a valid commitment to the sum of the underlying values and blinding factors:
\begin{equation}
    c_1 \odot c_2 = \mathsf{Commit}(x_1 + x_2, r_1 + r_2).
\end{equation}
These schemes must satisfy two fundamental security properties:
\begin{itemize}[leftmargin=*]
    \item \textbf{(Binding)} The committer cannot find a different value $x'\neq x$ and blinding factor $r'$ such that $c = \sf{Commit}\left(x', r'\right)$.
    \item \textbf{(Hiding)} The commitment $c$ leaks no information about the underlying value $x$.
\end{itemize}

\newcommand{\privin}{x_\text{priv}}
\newcommand{\pubin}{x_\text{pub}}

With the commitments established, a \emph{ZKP} is a protocol
\begin{equation}
    \Pi := \left(\prover\left(\privin, r, \pubin, y, c\right) \leftrightarrow \verifier\left(\pubin, y, c\right)\right)\label{eq:zkp}
\end{equation}
where a prover $\prover$ demonstrates to a semi-honest verifier $\verifier$ that a given computation is correct, without revealing the private input $\privin$ (e.g., clients' sensitive databases) or the blinding factor $r$~\cite{DBLP:conf/crypto/BellareG92, DBLP:conf/stoc/GoldwasserMR85, DBLP:conf/crypto/FiatS86, DBLP:conf/crypto/GoldreichMW86, DBLP:journals/ftsec/Thaler22}. Specifically, the prover demonstrates that for a public input $\pubin$ (e.g., verifier-chosen input known to both parties), public output $y$, and public commitment $c$, it holds simultaneously that $y = f\left(\privin, \pubin\right)$ and $c = \mathsf{Commit}\left(\privin, r\right)$. 

Equation~\eqref{eq:zkp} represents the protocol an honest prover follows. For a generic, potentially malicious prover $\prover$, we denote the execution as $\prover \leftrightarrow \verifier\left(\pubin, y, c\right)$, $\verifier\left(\prover, \pubin, y, c\right)$, or $\Pi\left(\prover, \pubin, y, c\right)$. Given a security parameter $\kappa \in \N$, a valid ZKP scheme satisfies:
\begin{itemize}[leftmargin=*]
    \item \textbf{(Completeness)} For honest $\prover$ that fully adheres to $\Pi$, and any $\left(\privin, r\right)$ where $y = f\left(\privin, \pubin\right)$ and $c = \mathsf{Commit}\left(\privin, r\right)$, $\verifier$ outputs $1$ (accepts) after executing $\Pi$.
    \item \textbf{(Soundness)} If no such $\left(\privin, r\right)$ exists, then for any $\prover$, the probability that $\verifier$ outputs $1$ is negligible in $\kappa$.
    \item \textbf{(Zero-knowledge)} Intuitively, the verifier learns absolutely nothing about the private input $\privin$ or the blinding factor $r$ beyond the fact that the proven statement is true. Formally, a PPT simulator $\simulator$ exists such that, for an honest prover with a valid private input $\privin$ and blinding factor $r$, $\simulator\left(\pubin, y, c\right)$ is computationally indistinguishable from $\verifier$'s view in $\Pi$.
\end{itemize}

While ensuring faithful execution and privacy, standard soundness is insufficient for randomized DP mechanisms with many valid outputs. Since DP mechanisms are randomized, many outputs (often the entire output space) are possible for each input. Even if $\prover$ lacks the valid private input $\privin$ (e.g., the clients' sensitive database), the mere existence of some $\privin$ trivializes standard soundness and fails to imply the correct sampling distribution. To prevent malicious output biasing, we additionally require:
\begin{itemize}[leftmargin=*]
    \item \textbf{(Knowledge Soundness)} Acceptance implies that $\prover$ \textbf{knows} the private input $\privin$ and blinding factor $r$ satisfying the aforementioned relations\footnote{We refer readers to standard cryptography textbooks~\cite{DBLP:books/cu/Goldreich2001,DBLP:journals/ftsec/Thaler22} for the formal definition of knowledge, i.e., knowing certain value(s).}. Due to the binding property of the commitment scheme, the known value of $\privin$ is unique.
\end{itemize}

\subsubsection{Secret-sharing} \label{sec:prelim-crypto-ss}

A $(t, n)$-threshold secret-sharing scheme~\cite{DBLP:journals/cacm/Shamir79, DBLP:conf/mark2/Blakley79} divides a secret $x$ into $n$ shares as $\left(\secsh{x}_i\right)_{i=0}^{n-1} \gets \mathsf{SecretShare}(x)$. Such a scheme guarantees that any subset of $t$ or more shares can be used to perfectly reconstruct the original secret $x$, while any subset of $t-1$ or fewer shares reveals absolutely no information about $x$.

We primarily utilize additive secret-sharing ($t=n$)~\cite{DBLP:books/cu/Goldreich2001,DBLP:books/crc/KatzLindell2014}, where all $n$ shares are required for reconstruction. For a secret $x \in \F$, the first $n-1$ shares $(\secsh{x}_0, \dots, \secsh{x}_{n-2})$ are sampled uniformly at random from $\F$, and the final share is computed as $\secsh{x}_{n-1} = x - \sum_{i=0}^{n-2} \secsh{x}_i$. This ensures the shares sum exactly to $x$, while possession of fewer than $n$ shares reveals no information. The multi-dimensional version applies this scheme element-wise.

Crucially, this additive sharing mechanism integrates naturally with our homomorphic commitments (Section~\ref{sec:prelim-crypto-zkp}). If a party additively secret-shares a value $x$ into $(\secsh{x}_i)_{i=0}^{n-1}$ and its random blinding factor $r$ into $(\secsh{r}_i)_{i=0}^{n-1}$, they can independently commit to each share as $\secsh{c}_i \gets \mathsf{Commit}(\secsh{x}_i, \secsh{r}_i)$. Due to the additive homomorphism of the commitment scheme, composing these individual share commitments with $\odot$ yields a valid commitment to the original secret $x$ with the combined blinding factor $r$:
\begin{equation}
    \bigodot_{i=0}^{n-1} \secsh{c}_i = \mathsf{Commit}\left(\sum_{i=0}^{n-1} \secsh{x}_i, \sum_{i=0}^{n-1} \secsh{r}_i\right) = \mathsf{Commit}(x, r).
\end{equation}
\section{Privacy Definition} \label{sec:def}
In this section, we describe the problem setup for a distributed mechanism in the client-server setting and identify the challenges in making the mechanism verifiable to an external data analyst. We then present our solution framework to address these challenges with formal security and privacy guarantees.

\subsection{Client-Server-Verifier Setup}\label{sec:def-setup}

The setup considered in this study involves $n_\cli \geq 1$ \textbf{potentially malicious clients}, $n_\ser \geq 1$ \textbf{potentially malicious servers}, and a \textbf{semi-honest} data analyst (who also acts as a \textbf{verifier}). As a motivating example, consider a financial institution setting: the clients $\vec{\cli} = \left(\cli_j\right)_{j \in \left[n_\cli\right]}$ are individual account holders, each holding a local transaction history $D_j \in \D$; the servers $\vec{\ser} = \left(\ser_i\right)_{i \in \left[n_\ser\right]}$ are internal bank servers distributed across departments and branches, responsible for processing these data under a distributed differentially private mechanism; and the data analyst $\verifier$ is an external auditor interested in verifying the correctness and integrity of the reported result.

Each client $\cli_j$ secret-shares $D_j$ with a subset of servers indexed by $I_j$, which limits the risk of individual data leakage unless a large number of servers collude. Each server $\ser_i$ aggregates the shares received from clients in $J_i := \left\{j: i \in I_j\right\}$, executes the designated computation, and submits the result to $\verifier$. Note that in central DP, $n_\ser = 1$ and $I_j = \left\{0\right\}$ for each $j$, so the secret-sharing scheme is trivial. Meanwhile, in generic secret-sharing-based distributed DP mechanisms where each client secret-shares to all servers~\cite{DBLP:conf/uss/BohlerK20, DBLP:conf/ccs/BohlerK21, DBLP:journals/tdsc/GoryczkaX17, DBLP:conf/nips/HeikkilaLKSTH17}, we have $I_j = [n_\ser]$ for each $j$ and $J_i = [n_\cli]$ for each $i$. To unify these common cases, we assume that for any pair of clients $\cli_j$ and $\cli_{j'}$, either $I_j = I_{j'}$ or $I_j \cap I_{j'} = \emptyset$. By elementary combinatorics, this implies that for any pair of servers $\ser_i$ and $\ser_{i'}$, either $J_i = J_{i'}$ or $J_i \cap J_{i'} = \emptyset$. For any set of servers $I$ that have received secret shares from the same set of clients, we use $J_I$ to denote these clients; conversely, for any set of clients $J$ that secret-share data with the same set of servers, we use $I_J$ to denote those servers.

\begin{figure}[!htbp]
    \centering
    \includegraphics[width=\linewidth]{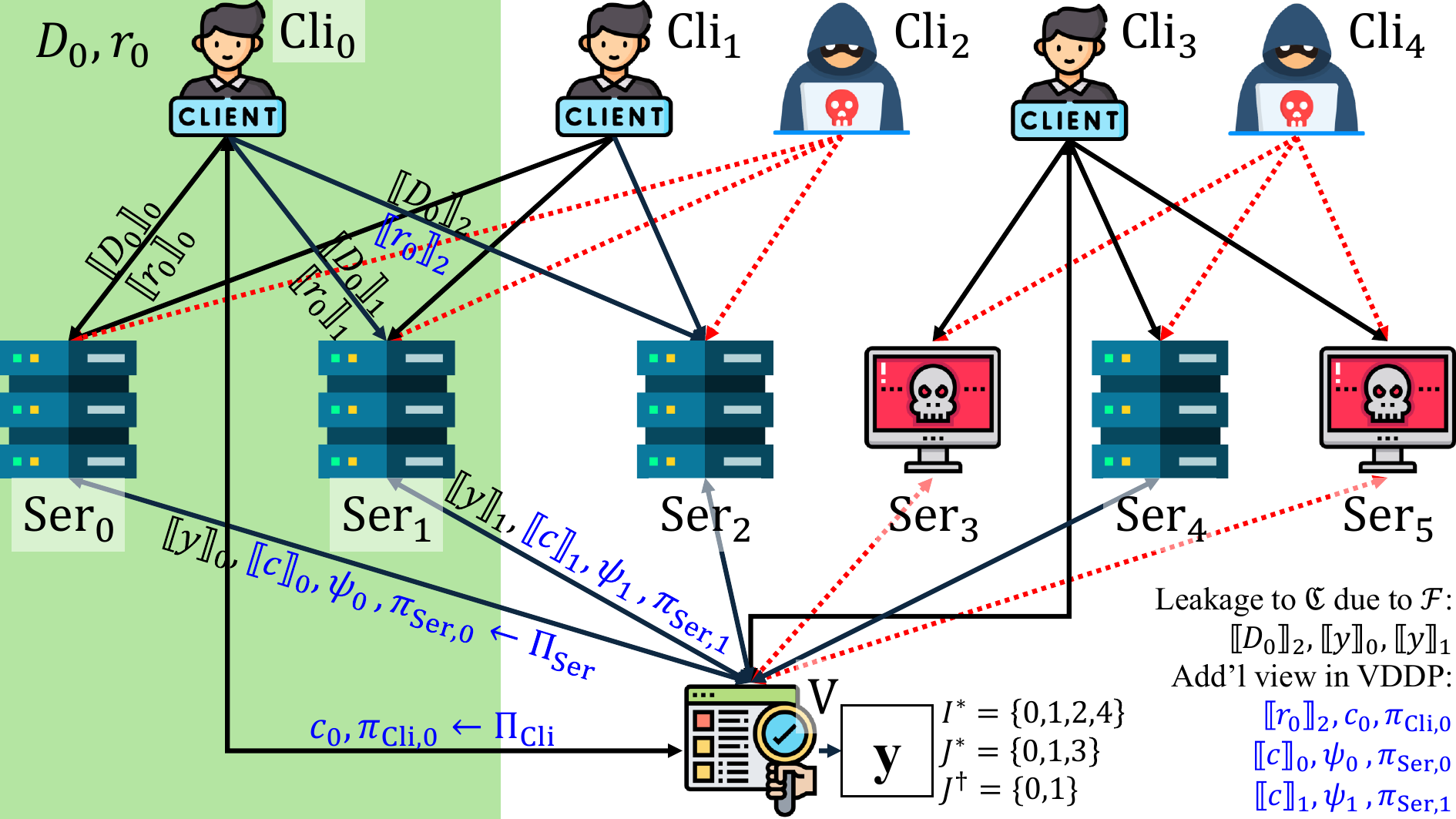}
    \caption{Problem setup of VDDP ($n_\cli = 5, n_\ser = 6$), targeting $\cli_0$'s data. All parties except $\cli_0$, $\ser_0$ and $\ser_1$ form a collusion $\mathfrak{C}$ to gain additional information on $D_0$, although some parties in $\mathfrak{C}$ are semi-honest and adhere to the protocol. All transmissions within and leaving the green region (solid black lines) are honestly computed but may cause information leakage and enable $\mathfrak{C}$ to compute additional information about $D_0$. The additional leakages under VDDP (commitments and proofs) are highlighted in blue. The transmissions from malicious clients (red dotted arrows) and servers may be adversarially reverse-engineered.}
    \label{fig:setting}
\end{figure}

\begin{Exm}
    In Figure~\ref{fig:setting}, $\cli_0$--$\cli_2$ are connected to $\ser_0$--$\ser_2$; $\cli_3$ and $\cli_4$ are connected to $\ser_3$--$\ser_5$. Consequently, $I_0=I_1=I_2=\left\{0,1,2\right\}$ and $I_3=I_4=\left\{3,4,5\right\}$. Meanwhile, $J_0=J_1=J_2=\left\{0,1,2\right\}$ and $J_3=J_4=J_5=\left\{3,4\right\}$. These server subsets $I_j$ and client subsets $J_i$ are either identical or disjoint.
\end{Exm}

\textbf{Aligned with the threat model of prior work under the same setup}~\cite{BC23}, we assume that clients or servers are potentially malicious. That is, clients may submit malformed or fabricated records $D_j \notin \D$ (e.g., transactions that are inconsistent with balance records or dated to the year 221 BC), while servers may alter the computation to weaken privacy guarantees or fabricate plausible-looking outputs to pass audits. Note that the scale of malicious behaviours could be arbitrary and difficult to control, ranging from individual malicious clients or servers (e.g., a client acting alone to conceal its money-laundering behaviour) to a large colluding coalition (e.g., a systematic corruption of the entire banking system involving all departments and branches collusively concealing illegal activities). Therefore, we do not make any assumption about the maximal proportion of malicious or colluding clients and servers.

On the other hand, the data analyst (verifier) $\verifier$ is modeled as \emph{semi-honest}. That is, they follow the audit protocol and report the outcome (pass or fail) honestly, but may abuse their access to attempt to infer sensitive transaction data from individuals. For example, the auditor has no incentive to intentionally misreport the honesty of the clients and servers since such misreporting would undermine the audit's credibility and mislead the auditor itself to an incorrect conclusion. However, they could still be interested in obtaining complete or partial copies of clients' transaction records for better auditing capabilities and potentially illegitimate use. $\verifier$ is responsible for verifying both \textbf{1)}~the validity of the client databases (i.e., $D_j \in \D$) and \textbf{2)}~the correctness of the computations performed by the servers. Moreover, $\verifier$ may collaborate with clients and servers to extract additional information about other clients' databases from the outputs and proofs. For example, some clients or branches may be personally acquainted with the financial auditor, or have common interests with the latter, and therefore opt to collaborate.

\subsection{The VDDP Protocol}

To defend against protocol deviations arising from the aforementioned malicious behaviours of clients and servers, we introduce the \emph{Verifiable Distributed Differential Privacy (VDDP)} protocol, which verifies the validity of clients' local databases and the correctness of the servers' computation, thereby guaranteeing the correctness of the result distribution received by the verifier. In particular, we aim to resolve the following challenges through the protocol design:

\begin{enumerate}[wide, labelwidth=!, labelindent=0pt, label=\textbf{Challenge~\arabic*.}, ref=\textbf{Challenge~\arabic*}]

\item \label{challenge:interference} Unlike in MPC protocols for DP mechanisms~\cite{DBLP:conf/eurocrypt/DworkKMMN06, DBLP:conf/ccs/ChampionSU19, DBLP:conf/ccs/WeiYFCW23}, there is no assumption that a certain portion (e.g., two-thirds or half) of the servers is honest. Instead, the honesty of the servers (and clients) is inspected by an external verifier. Additionally, the local databases of a subset of clients may be removed from the input upon the discovery of dishonest clients and servers to prevent them from affecting the verifier's result. Therefore, an MPC ideal functionality over a fixed secret-shared input database cannot be well defined~\cite{DBLP:conf/uss/BohlerK20,DBLP:conf/ccs/BohlerK21}. Instead, the verifier needs to identify the set of honest clients (indexed by $J^*$) and servers (indexed by $I^*$), and further, a subset of clients $J^\dagger \subseteq J^*$ such that the verifier's aggregated result is \textbf{correct with respect to the union of data from clients in $J^\dagger$}.

\begin{Exm}\label{exm:removal}
    In Figure~\ref{fig:setting}, $\cli_0$ and $\cli_1$ (with inputs $D_0$ and $D_1$) are both semi-honest and adhere to the protocol. Even though another dishonest client $\cli_2$ submits its data to the same set of servers $\left\{\ser_0, \ser_1, \ser_2\right\}$, the protocol may still continue on the input of $D_0$ and $D_1$ after the dishonesty of $\cli_2$ is revealed. Meanwhile, for $\cli_3$, as too many servers processing its data ($\ser_3$ and $\ser_5$) are dishonest, it should be excluded from the protocol to preserve the integrity of the final result. In this case, $3 \in J^*$ but $3\notin J^\dagger$.
\end{Exm}

\item \label{challenge:collusion} The servers' computation is vulnerable to the collusion of the malicious servers. Therefore, enforcing the output distribution correctness of individual servers is insufficient; instead, the joint distribution of all server outputs must be verifiable as a security guarantee of the protocol. In particular, in adapting randomness to verifiable computation~\cite{DBLP:books/cu/Goldreich2001, BC23}, the servers' sampling processes must be uncorrelated to ensure the final output distribution's correctness and establish end-to-end security.

\begin{Exm} \label{exm:server-collusion}
    When applying the distributed Gaussian mechanism, two colluding servers $\ser_{i_1}$ and $\ser_{i_2}$ may sample a pair of mutually cancelling noise values $r$ and $-r$, while each value still follows the prescribed Gaussian distribution. Therefore, although the marginal output distributions of each server are correct, the cancellation of the noise nullifies the prescribed DP guarantee.
\end{Exm}

\item \label{challenge:privacy} Unlike conventional secure computations of DP mechanisms~\cite{DBLP:conf/uss/BohlerK20, DBLP:conf/ccs/BohlerK21, DBLP:conf/eurocrypt/DworkKMMN06, DBLP:conf/ccs/ChampionSU19, DBLP:conf/ccs/WeiYFCW23}, verifiable computation includes a proof of correct execution that may cause additional information leakage even when the protocol output inherits the underlying mechanism's privacy guarantees, thereby breaking the original DP guarantee. Previous studies use zero-knowledge (ZK) proofs~\cite{BC23, DBLP:conf/iclr/ShamsabadiTCBHP24} to prevent further leakage at this stage. However, as these are distinct privacy notions, it is worth exploring whether the ZK property of the proof is necessary and sufficient for achieving end-to-end DP. More importantly, due to the distributed nature of the setting, the privacy guarantee must be established with respect to each client's local database against other potentially colluding clients, servers, and the verifier.
\end{enumerate}

The design of the protocol aims to enable clients to secret-share their sensitive databases $D_j$, such that the servers can aggregate the secret shares from multiple clients as $\secsh{D}_i$ and compute a \emph{server function} $\mathcal{F}\left(\secsh{D}_i\right)$ over it. To facilitate verifiable computation over the randomized function, we apply a distributed randomness scheme that splits the randomness into the server seed $\sigma_i$ (kept secret by the server) and verifier seed $\phi_i$~\cite{BC23, DBLP:books/cu/Goldreich2001}, such that the server evaluates a deterministic function $f$ using $\sigma_i$ and $\phi_i$, sampled from distributions $\mathcal{P}\left(\sigma\right)$ and $\mathcal{P}\left(\phi\right)$, respectively, where

\begin{itemize}[leftmargin=*]
    \item For any fixed $\sigma_i$, the distribution of
    $f\left(\secsh{D}_i, \sigma_i, \phi_i\right)$ (over the randomness of $\phi_i$) is equivalent to $\mathcal{F}(\secsh{D}_i)$.
    \item For any fixed $\phi_i$, the distribution of
    $f\left(\secsh{D}_i, \sigma_i, \phi_i\right)$ (over the randomness of $\sigma_i$) is equivalent to $\mathcal{F}(\secsh{D}_i)$.
\end{itemize}

This distributed randomness scheme ensures that the server cannot bias the distribution by controlling $\sigma_i$. Meanwhile, the verifier gains no additional information about input $\secsh{D}_i$ through knowledge of $\phi_i$.

\begin{Exm}[Construction with LPRF]\label{eg:dist-rand-lprf}
    Consider any ideal server function $\mathcal{F}^*$ such that there exists a deterministic arithmetic circuit $C$ whereby, for any input $\secretshare{D}_i$, $\mathcal{F}^*\left(\secretshare{D}_i\right)$ has the same distribution as $C\left(\secretshare{D}_i, \mathbf{b}_i\right)$, where $\mathbf{b}_i\sample \bin^B$.

    To efficiently generate and verify a large number of coins, we employ the LPRF and consider the real server function $\mathcal{F}$ where $\mathcal{F}\left(\secretshare{D}_i\right)$ has the same distribution as $C\left(\secretshare{D}_i, \sf{LPRF}\left(s_i\right)\right)$, where $s_i\sample \F$ such that $\mathcal{F}\left(\secretshare{D}_i\right)$ is computationally indistinguishable from $\mathcal{F}^*\left(\secretshare{D}_i\right)$. Hence, a distributed randomness scheme for $\mathcal{F}$ can be constructed as
    \begin{equation}
        f\left(\secretshare{D}_i, \sigma_i, \phi_i\right) := C\left(\secretshare{D}_i, \sf{LPRF}\left(\sigma_i + \phi_i \right)\right), \label{eq:rd-lprf}
    \end{equation}
    where $\sigma_i$ and $\phi_i$ are both sampled from the uniform distribution over $\F$. Fixing either $\sigma_i$ or $\phi_i$ and sampling the other uniformly ensures that $\sigma_i + \phi_i$ is uniformly distributed over $\F$, which yields the same output distribution as $\mathcal{F}\left(\secretshare{D}_i\right)$.
\end{Exm}

To address the aforementioned challenges, we design the VDDP protocol $\Pi$ that is compatible with the client-server-verifier setup outlined in Section~\ref{sec:def-setup}.

The protocol begins by having parties prepare and commit to the following public and private inputs:

\begin{itemize}[leftmargin=*] 
    \item $\left(\secretshare{D_j}_i, \secretshare{r_j}_i\right)_{i \in I_j} \sample \sf{SecretShare}\left(D_j, r_j\right)$: each client $\cli_j$ secret-shares its local database $D_j$ and the blinding factor $r_j$ to the corresponding servers\footnote{We assume that the clients and servers agree on the commitments to the shares (and therefore the underlying committed values), which can be achieved using digital signatures on the commitments~\cite{bls, rsa, elgamal, eddsa, DBLP:books/crc/KatzLindell2014}}.
    \item $\secretshare{c_j}_i \gets \sf{CommitShare}\left(\secretshare{D_j}_i, \secretshare{r_j}_i\right)$: each client $\cli_j$ publishes the commitment to each secret share $\secretshare{D_j}_i$ with secret-shared blinding factor $\secretshare{r_j}_i$, where $i \in I_j$. By leveraging the homomorphism of the commitment scheme, all parties can compute a valid commitment $c_j$ to $\left(D_j, r_j\right)$.
    \item $\sigma_i\sample \mathcal{P}\left(\sigma\right)$: each server $\ser_i$ chooses its server seed $\sigma_i$.
    \item $\psi_i \gets \sf{CommitSeed}\left(\sigma_i, \rho_i\right)$: each server $\ser_i$ publishes the commitment to $\sigma_i$ with blinding factor $\rho_i$. 
    \item $\phi_i \sample \mathcal{P}\left(\phi\right)$: the verifier samples the public coins independently for each server. This step is performed after $\sf{CommitSeed}$ to prevent reverse-engineered choices of $\sigma_i$.
\end{itemize}

With these private and public inputs in place, the VDDP protocol involves the following interactions among the parties:

\begin{itemize}[leftmargin=*]
    \item $\Pi_{\cli}:=\cli_j\leftrightarrow\verifier\left(c_j\right)$: each client $\cli_j$ proves to $\verifier$ the validity of its database (i.e., $D_j\in \D$), with respect to private input $D_j$, blinding factor $r_j$, and public commitment $c_j$. The verifier announces the indices of the accepted clients $J^*$.
    \item $\left(\secretshare{D}_i, \secretshare{r}_i\right)\gets \sf{AggrShare}\left(\left(\secretshare{D_j}_i, \secretshare{r_j}_i\right)_{j\in J^*\cap J_i}\right)$: each server aggregates the secret shares from the associated clients that are accepted by $\verifier$. By leveraging the homomorphism of the commitment scheme, all parties can compute a valid commitment $\secsh{c}_i$ of $\left(\secretshare{D}_i, \secretshare{r}_i\right)$.
    \begin{Exm}
        $\sf{AggrShare}$ can be instantiated as concatenation for generic queries, and addition for histogram queries. 
    \end{Exm}
    \item $\secsh{y}_i\gets f\left(\secretshare{D}_i, \sigma_i, \phi_i\right)$: each server publishes the secret-shared output $\secsh{y}_i$, computed on the secret shares received from the clients that have been accepted.
    \item $\Pi_{\ser}:=\ser_i\leftrightarrow\verifier\left(\phi_i, \secsh{y}_i, \secsh{c}_i, \psi_i\right)$: each server $\ser_i$ proves to $\verifier$ the correctness of the computation over $f$, with respect to the private inputs $\secsh{D}_i$ and $\sigma_i$, blinding factors $\secsh{r}_i$ and $\rho_i$, public input $\phi_i$, public output $\secsh{y}_i$, and public commitments $\secsh{c}_i$ and $\psi_i$. The verifier announces the indices of the accepted servers $I^*$.
    \item $\left(y, J^\dagger\right)\gets \sf{AggrOut}\left(\left(\secsh{y}_i\right)_{i\in I^*}, I^*, J^*\right)$: the verifier aggregates the secret-shared outputs from the accepted servers and identifies the subset of clients $J^\dagger$ that contributes to the final result. In particular, accepted clients ($j\in J^*$) may still be excluded from $J^\dagger$ because of the dishonesty of associated servers (see Example~\ref{exm:removal}).
    \begin{Exm} \label{exm:aggrout}
        Consider the protocol using additive secret sharing where each server adds a copy of the noise to its output $\secsh{y}_i$, where the DP guarantee holds against the data analyst and all other colluding servers~\cite{BC23}. In this scenario, the data $D_j$ of a client $\cli_j$ contributes to the final result $y$ only if all associated servers are identified as honest by the verifier. Therefore, the verifier can identify the contributing clients as $J^\dagger\gets J^*\cap\left\{j: I_j \subseteq I^*\right\}$, and the final result can only be aggregated from the servers associated with the contributing clients, as $y\gets \sum_{i\in I_{J^\dagger}} \secsh{y}_i$. In particular, in Figure~\ref{fig:setting} and Example~\ref{exm:removal}, $J^\dagger = \left\{0,1\right\}$ despite $J^* = \left\{0, 1, 3\right\}$.
    \end{Exm}
\end{itemize}

\subsection{Security Guarantees}

In response to \ref{challenge:interference} and \ref{challenge:collusion}, we formalize the security guarantees of the VDDP protocol. In particular, regarding \ref{challenge:interference}, as exemplified by Example~\ref{exm:removal}, honest clients may still be ultimately removed because of the malicious behaviour of the associated servers. Therefore, unlike ZKP protocols with a single prover and a single verifier, where the honest prover is always accepted, the completeness of the VDDP protocol (Definition~\ref{def:vddp-completeness}) captures the scenario in which the aforementioned clients are potentially removed.

\begin{Def}[$\theta$-completeness of VDDP Protocol]\label{def:vddp-completeness}
    A VDDP protocol $\Pi$ is \emph{$\theta$-complete} iff for any honest client $\cli_j$ with at least $|I_j| - \theta$ honest servers indexed by $I_j$, we have $j \in J^\dagger$ with probability $1$.
\end{Def}

By definition, if $\Pi_\cli$ is complete, then an honest client $\cli_j$ with a valid local database $D_j \in \D$ is included in $J^*$ with probability 1. Similarly, an honest server $\ser_i$ that performs the correct computation and generates the proof is also included in $I^*$ with probability 1. Nevertheless, as stated in Proposition~\ref{prop:def-comp}, the $\sf{AggrOut}$ may add another layer that filters the contributing clients.

\begin{Prop}\label{prop:def-comp}
    A VDDP protocol $\Pi$ is $\theta$-complete if $\Pi_\cli$ and $\Pi_\ser$ are complete and, for any client $\cli_j$, $\sf{AggrOut}$ outputs $j \in J^\dagger$ whenever $j\in J^*$ and $\abs{I_j\cap I^*}\geq \abs{I_j} - \theta$.
\end{Prop}

As illustrated by Example~\ref{exm:aggrout}, the design of $\sf{AggrOut}$ is closely related to the secret-sharing scheme used. In particular, $\theta \leq \abs{I_j} - t$ in any $\left(t, \abs{I_j}\right)$-threshold secret-sharing scheme: otherwise, with only $\abs{I_j\cap I^*} < t$ pieces of secret-shared outputs, the verifier may not aggregate any useful information computed on $D_j$. Therefore, to increase fault tolerance and strike a balance between the privacy requirement and contribution to the final output, clients may secret-share their input databases to a subset instead of all servers. This strategy decreases the chance of the ``unfortunate'' event that one or a few malicious servers cause their data to be excluded from the protocol, despite their own honesty and that of all other servers handling their data.

Meanwhile, to prevent collusion among servers through correlated sampling and resolve \ref{challenge:collusion}, we further identify that the distributed randomness strategy~\cite{DBLP:books/cu/Goldreich2001,BC23} used to compute the server functions $\mathcal{F}$ is resistant to collusion. In particular, as long as the $\phi_i$ values are sampled independently (which is guaranteed by the semi-honesty of the verifier), the servers' outputs remain uncorrelated.

\begin{Lem}\label{lem:dist-rand-collusion}
    For any joint distribution over $\left(\sigma_0, \sigma_1, \dots, \sigma_{n_\ser - 1}\right)$, if each $\phi_i$ where $i \in [n_\ser]$ is sampled i.i.d. from $\phi_i\sample \mathcal{P}\left(\phi\right)$, then the joint distribution of $\left(f(\secsh{D}_i, \sigma_i, \phi_i)\right)_{i\in [n_\ser]}$ is equivalent to the Cartesian product of the distributions of all $\mathcal{F}(\secsh{D}_i)$, i.e., $\bigotimes_{i \in [n_\ser]}\mathcal{F}\left(\secsh{D}_i\right)$. 
\end{Lem}

\begin{proof}
    Conditioned on any $\left(\sigma_0, \sigma_1, \dots, \sigma_{n_\ser - 1}\right)$, the definition of $f$ and the independence of $\phi_i$ for each $i$ imply that the joint distribution of the server outputs $\left(\secsh{y}_i\right)_{i\in[n_\ser]}$ is equivalent to $\bigotimes_{i \in [n_\ser]}\mathcal{F}\left(\secsh{D}_i\right)$. Therefore, by the law of total probability, the equivalence also holds for any joint distribution of $\left(\sigma_i\right)_{i\in[n_\ser]}$.
\end{proof}

Furthermore, to ensure that the VDDP protocol safeguards the correctness of the computation over the prescribed DP mechanism, we adapt the notion of knowledge soundness of the ZKP protocol (Definition~\ref{def:vddp-ks}). In particular, due to the randomized nature of DP mechanisms, an output could correspond to multiple possible input databases; therefore, the weaker notion of soundness is insufficient for ensuring that the mechanism has been faithfully executed with respect to one set of fixed secret input databases.

\begin{Def}[Knowledge soundness of VDDP Protocol]\label{def:vddp-ks}
    A VDDP protocol $\Pi$ is \emph{knowledge sound} iff, except with negligible probability, the clients and servers collectively know a unique set of:
    \begin{itemize}[leftmargin=*]
        \item For each $j\in J^\dagger$: local database $D_j$ and blinding factor $r_j$ that are valid openings of $c_j$.
        \item For each $i\in I_{J^\dagger} := \bigcup_{j\in J^\dagger}I_j$: the aggregated secret shares of databases $\secsh{D}_i$ and blinding factor $\secsh{r}_i$ that are valid openings of $\secsh{c}_i$; and server seed $\sigma_i$ and blinding factor $\rho_i$ that are valid openings of $\psi_i$, such that $\secretshare{y}_i = f\left(\secsh{D}_i, \sigma_i, \phi_i\right)$.
    \end{itemize}
\end{Def}

Although the malicious clients and servers may also collude during the proof generation, by treating them collectively as one PPT adversary, the knowledge soundness of the $\Pi_\cli$ and $\Pi_\ser$ sub-protocols still implies that they collectively know a (unique) value of each entity in Definition~\ref{def:vddp-ks}:

\begin{Prop}
    A VDDP protocol $\Pi$ is knowledge sound if $\Pi_\cli$ and $\Pi_\ser$ are knowledge sound.
\end{Prop}

\subsection{Privacy Guarantees}

In response to \ref{challenge:privacy}, we establish an end-to-end DP guarantee that accounts for privacy leakage during the execution of the VDDP protocol. Due to the distributed nature of the protocol, the notion of DP can only be defined on a per-client basis by considering a pair of neighbouring local databases $D_j$ and $D_j'$ for each client $\cli_j$. Moreover, if all servers processing $D_j$ were to \emph{collude}, achieving any degree of DP would be impossible, as $D_j$ could be fully recovered via the secret-sharing scheme. Therefore, DP can only be achieved for the clients when the number of servers that are not fully honest is bounded. Denoting all fully honest servers in $I_j$ as $H_j$, we treat all other clients (except $\cli_j$), all servers other than those in $H_j$, and $\verifier$ collectively as a collusion $\mathfrak{C}$.

\begin{Def}[End-to-end Privacy Guarantees of VDDP Protocol]\label{def:vddp-dp} 
    A VDDP protocol $\Pi$ is $\left(\tau, \epsilon\right)$-private if for any honest client $\cli_j$, any set of honest servers $H_j \subseteq I_j$ such that $\abs{H_j} \geq \abs{I_j} - \tau$, the view of all remaining parties $\mathfrak{C}$ (consisting of $\verifier$ and all clients and servers except $\cli_j$ and $H_j$), $\sf{View}_\Pi^\mathfrak{C}\left(D_j\right)$ is $\epsilon$-IND-CDP in $D_j$.
\end{Def}

With the application of a $\left(t, \abs{I_j}\right)$-threshold secret sharing scheme, if $\abs{I_j} - \abs{H_j}\geq t$, then $\mathfrak{C}$ can fully recover the secret-shared database $D_j$, which provides no meaningful privacy guarantee. Therefore, in Definition~\ref{def:vddp-dp}, $\tau < t$.

While the commitments and the zero-knowledge sub-protocols $\Pi_\cli$ and $\Pi_\ser$ cause no additional information leakage to $\mathfrak{C}$, the privacy cost of the server outputs needs to be accounted for. For the honest servers indexed by $H_j$, the joint distribution of the server outputs $\left(f\left(\secsh{D}_i, \sigma_i, \phi_i\right)\right)_{i\in H_j}$ has the same distribution as $\bigotimes_{i\in H_j}\mathcal{F}\left(\secsh{D}_i\right)$ because $\sigma_i$ is honestly sampled for each $i$. However, the values of $\secsh{D}_i$, and hence the outputs, depend not only on $D_j$ but also on the exact values of the secret shares $\secsh{D_j}_i$, as well as the secret shares from other clients in $\mathfrak{C}$ and a subset of honest clients $J^*$ that may be chosen adaptively. Additionally, although the secret shares transmitted to $\mathfrak{C}$ alone leak no information, they are correlated with the server outputs. Therefore, when capturing privacy leakage through the server outputs, we consider the joint distribution of both the server outputs and the secret shares:

\begin{Def}\label{def:server-func-dp}
    The server function $\mathcal{F}$ is $\left(\tau, \epsilon\right)$-private if for any honest client $\cli_j$, any set of honest servers $H_j \subseteq I_j$ such that $\abs{H_j} \geq \abs{I_j} - \tau$, any set of secret shares transmitted from $\mathfrak{C}$ to $H_j$, i.e., $\secsh{D_{j'}}_i$ for $i \in H_j$ and $j'\in J_i\backslash\left\{j\right\}$, and any $J^*$ such that $j\in J^*$, for the aggregated $\secsh{D}_i$, the joint distribution of \begin{equation}
        \left(\left(\mathcal{F}\left(\secsh{D}_i\right)\right)_{i\in H_j}, \left(\secsh{D_j}_{i'}\right)_{i'\in I_j\backslash H_j}\right)
    \end{equation} is $\epsilon$-IND-CDP with respect to $D_j$.
\end{Def}

\begin{Prop}\label{prop:zk-dp}
    A VDDP protocol $\Pi$ is $\left(\tau, \epsilon\right)$-private if it is instantiated with hiding commitment schemes, zero-knowledge $\Pi_\cli$ and $\Pi_\ser$, and the server function $\mathcal{F}$ is $\left(\tau, \epsilon\right)$-private.
\end{Prop}

\begin{Exm}\label{exm:histogram}
    Consider the histogram query where each client $\cli_j$ holds a database $D_j\in \bin^{N\times d}$, i.e., $N$ records, with each row being a one-hot representation of membership in one of the $d$ bins. Denote the histogram of $D_j$ as $\mathbf{x}_j\in \N^d$. 
    
    To compute the histogram over all clients' data, each server $\ser_i$ aggregates $\secsh{D}_i\gets \sum_{j\in J^*\cap J_i}\secsh{D_j}_i$, and computes the column-wise sum as $\secsh{\mathbf{x}}_i$, which is a valid secret share of $\sum_{j\in J^*\cap J_i}\mathbf{x}_j$. To protect the privacy of the associated clients, each server adds a copy of noise $\mathbf{r}_i$ to the output as $\secsh{\mathbf{y}}_i\gets \secsh{\mathbf{x}}_i + \mathbf{r}_i$. Therefore, the final aggregated result is $\mathbf{y} = \sum_{j\in J^\dagger}\mathbf{x}_j + \sum_{i\in I_{J^\dagger}} \mathbf{r}_i$.

    Observe that, for any strict subset of honest servers $H_j \subset I_j$, the joint distribution of the honest servers' outputs and the secret shares from $\cli_j$ to $\mathfrak{C}$ is \begin{equation}
        \left(\left(\secsh{\mathbf{x}_j}_i + \sum_{j'\in J^*\cap J_i\backslash\left\{j\right\}}\secsh{\mathbf{x}_{j'}}_i + \mathbf{r}_i\right)_{i\in H_j}, \left(\secsh{D_j}_{i'}\right)_{i'\in I_j\backslash H_j}\right).
    \end{equation} For fixed $\secsh{D_{j'}}_i$ and $J^*$, the privacy budget $\epsilon$ is equivalent to that of the central-DP setting with the additive noise of $\sum_{i\in H_j}\mathbf{r}_i$. Therefore, the server function and the corresponding VDDP protocol are $\left(\abs{I_j} - \abs{H_j}, \epsilon\right)$-private.
\end{Exm}

\subsubsection{Relationship between ZKP and VDDP}

Although zero-knowledge sub-protocols $\Pi_\cli$ and $\Pi_\ser$ suffice to construct private VDDP protocols, we further argue that this property is not necessary. We construct a counterexample in which the non-ZK $\Pi_\cli$ leaks only a bounded amount of information about the input database $D_j$. This finding implies that a wider range of privacy-preserving protocols exists beyond ZKPs, opening the door to potentially more efficient constructions of VDDP protocols in future work.

To construct a concrete counterexample, we consider a simplified counting-query scenario used in the prior work of VDBM under the same setting~\cite{BC23}. In VDBM, each $\cli_j$ holds a bit $x_j\in\bin$ as its local data, where $0$ and $1$ are considered neighbouring inputs, and secret-shares it with the servers. Each server adds a copy of the binomial noise to achieve privacy protection with privacy budget $\epsilon$. $\Pi_\cli$ is instantiated as a folklore complete and knowledge-sound ZKP protocol $\Pi_\sf{Bin}$ that proves $x_j\in\bin$~\cite{DBLP:conf/crypto/CramerDS94, DBLP:conf/eurocrypt/Damgard00, DBLP:journals/ftsec/Thaler22}. However, by biasing the randomness used by the prover (client), we can construct a variation of $\Pi_\sf{Bin}$ where the verifier's view is $\epsilon'$-DP for any $\epsilon' > 0$, increasing the total privacy leakage to $\epsilon + \epsilon'$.

\begin{Prop}
    For any $\left(\tau,\epsilon\right)$-private VDBM protocol and any $\epsilon' > 0$, there exists a variation with the same server function and non-ZK $\Pi_\cli$ that is $\left(\tau,\epsilon + \epsilon'\right)$-private.
\end{Prop}

The full technical construction (modified protocol and full proof of the corresponding theorem) is available in Appendix~\ref{app:nonzk-vdbm-variant}.
\section{Verifiable Distributed Discrete Laplace Mechanism (VDDLM)} \label{sec:vddlm}

In this section, we develop VDDLM for counting queries. Compared with VDBM, VDDLM incurs substantially lower computational overhead while closely preserving the privacy--utility trade-off of the original DLM~\cite{DBLP:journals/siamcomp/GhoshRS12, DBLP:conf/innovations/BalcerV18}. Similar to the setting of VDBM~\cite{BC23}, clients distribute their local counts to all servers using additive secret-sharing. Each server aggregates these secret shares and adds an independent, element-wise copy of discrete Laplace noise to each output dimension. The data analyst, acting as the verifier, then aggregates the server outputs to reconstruct the final result. We first present the sampling-circuit formulation for VDDLM and then analyze privacy and complexity. Despite fixed-point numerical errors, the implemented sampling circuit remains a close approximation of the DLM, admits tight privacy bounds, and yields an asymptotically improved circuit complexity.

Specifically, VDDLM instantiates the histogram setting in Example~\ref{exm:histogram}. Therefore, it suffices to design a verifiable sampling circuit for additive noise in the central setting, which automatically translates to the VDDP setting with the same privacy budget through the framework in Section~\ref{sec:def}.

\subsection{Additive Noise Sampling Circuit} \label{sec:vddlm-circuit}

The core design goal of VDDLM is to realize a sampling circuit that is simultaneously faithful to the target DLM and efficient to verify inside a proof system. Specifically, for the target discrete Laplace distribution $\sf{Lap}_\Z\left(\frac{\Delta}{\epsilon^\dagger}\right)$, where $\Delta$ is the sensitivity and $\epsilon^\dagger$ is the privacy budget, it is not sufficient to show that the distribution of the constructed circuit's output is close to $\sf{Lap}_\Z\left(\frac{\Delta}{\epsilon^\dagger}\right)$; it is also vital to guarantee that the sampled additive noise incurs a similar privacy budget, $\hat{\epsilon} \approx \epsilon^\dagger$. Additionally, to account for approximation errors that cannot be captured by the difference between $\epsilon^\dagger$ and $\hat{\epsilon}$ (e.g., unavoidable truncation errors due to finite-precision arithmetic), we tolerate a negligibly small failure probability $\delta^\dagger$ in the privacy guarantee.

\begin{algorithm}[!htbp]
\caption{Approximation circuit for discrete Laplace noise sampling in VDDLM}
\label{alg:vddlm_circuit}
\autofit{\begin{algorithmic}[1]
\Require Private Bernoulli bits $\left\{b_{c,i,k}\right\}_{c \in \{+, -\}, i\in\range{\gamma}, k\in\range{\nu_i}}$
\Require Public parameters $\gamma, \bs{\nu}, \bs{\tau}$ generated from Algorithm~\ref{alg:param_selection}
\Ensure Additive noise $r$
\For{$c \in \{+, -\}$} \label{line:vddlm-components-start}
    \State $X_c \gets 0$ \Comment{Positive/negative component, as Equation~\eqref{eq:dlap-plus-minus-decompose}}
    \For{$i = 0$ to $\gamma - 1$} \label{line:vddlm-gamma-loop}
        \State $V_{c,i} \gets \sum_{k=0}^{\nu_i-1} 2^k \cdot b_{c,i,k}$ \label{line:vddlm-v-compute}
        \State $B_{c,i} \gets \1\left[V_{c,i} < \tau_i\right]$ \label{line:vddlm-bernoulli} \Comment{Bernoulli samples, as Equation~\eqref{eq:geometric-decompose}}
        \State $X_c \gets X_c + 2^i \cdot B_{c,i}$ \label{line:vddlm-x-accumulate}
    \EndFor
\EndFor \label{line:vddlm-components-end}
\State $r \gets X_+ - X_-$ \label{line:vddlm-noise-construct} \Comment{Construct discrete Laplace noise}
\State \Return $r$
\end{algorithmic}}
\end{algorithm}

As outlined in Algorithm~\ref{alg:vddlm_circuit}, our construction of the sampling circuit utilizes the following two facts in probability theory:

\begin{itemize}[leftmargin=*]
    \item Given two mutually independent random variables $X_+, X_-$ that follow the geometric distribution, i.e., \begin{equation}
        \prob{X_+ = x} = \prob{X_- = x} \propto e^{-\frac{\epsilon^\dagger}{\Delta}x} \label{eq:dlap-plus-minus-decompose}
    \end{equation} for any non-negative integer $x$, $X_+ - X_-\sim \sf{Lap}_\Z\left(\frac{\Delta}{\epsilon^\dagger}\right)$.
    \item In binary representation $X_+ = \sum_{i\geq 0}2^iB_{+,i}$ (and symmetrically for $X_- = \sum_{i\geq 0}2^iB_{-,i}$), the bits $B_{+,i}$ are mutually independent~\cite{DBLP:conf/eurocrypt/DworkKMMN06, DBLP:conf/ccs/WeiYFCW23}, i.e., for any bit string $\left(b_i\right)_{i\geq 0}$, \begin{equation}
        \prob{X_+ = \sum_{i\geq 0}2^ib_i} = \prod_{i\geq 0}\prob{B_{+, i} = b_i}. \label{eq:geometric-decompose}
    \end{equation} Specifically, each $B_{+,i}$ (or $B_{-,i}$) follows a Bernoulli distribution with a success probability of $\prob{B_{+,i}=1} = \frac{1}{1 + \exp\left(2^i\frac{\epsilon^\dagger}{\Delta}\right)} := p_i^\dagger$.
\end{itemize}

In practice, evaluating an infinite sequence is impossible; therefore, we truncate the generation loop to the least significant $\gamma$ bits for both $X_+$ and $X_-$ (Line~\ref{line:vddlm-gamma-loop}). Furthermore, for each step $i$, we must approximate the ideal Bernoulli variables $B_{+,i}$ and $B_{-,i}$. We do this by reconstructing $\nu_i$-bit integers from our LPRF bitstream, denoted as $V_{+,i}$ and $V_{-,i}$ (Line~\ref{line:vddlm-v-compute}). The ideal probability $p_i^\dagger$ is quantized into a fixed-point representation $\hat{p}_i$, defining an integer threshold $\tau_i = \lfloor 2^{\nu_i}\hat{p}_i \rceil$. The in-circuit Bernoulli outcome is then constrained via a simple threshold comparison, $B_{c,i} = \1\left[V_{c,i} < \tau_i\right]$ (Line~\ref{line:vddlm-bernoulli}), which can be implemented via bit-decomposition~\cite{DBLP:journals/ftsec/Thaler22}. These Bernoulli samples are then assembled as $X_+$ and $X_-$ (Line~\ref{line:vddlm-x-accumulate}) and finally synthesized as the output noise $r$ (Line~\ref{line:vddlm-noise-construct}). In the full VDDLM protocol, this sampled noise is then added to the aggregated client shares, as in Example~\ref{exm:histogram}, to obtain the noisy output.

The choices of the circuit parameters $\gamma$, $\bs{\nu} = \left(\nu_i\right)_{i\in\range{\gamma}}$, and $\bs{\tau} = \left(\tau_i\right)_{i\in\range{\gamma}}$ are discussed in Section~\ref{sec:vddlm-param-selection}. We further demonstrate in Section~\ref{sec:vddlm-analysis} that the circuit construction in Algorithm~\ref{alg:vddlm_circuit} reduces the overhead compared with previous constructions, while incurring a controllable amount of additional privacy leakage from the combined numerical errors (e.g., quantization and truncation).

\subsection{Dynamic Parameter Selection} \label{sec:vddlm-param-selection}

\begin{algorithm}[!htbp]
\caption{Dynamic parameter selection for VDDLM}
\label{alg:param_selection}
\autofit{\begin{algorithmic}[1]
\Require Target privacy bounds $\epsilon^\dagger, \delta^\dagger$, sensitivity $\Delta$, error tolerance $\xi$
\Ensure Step count $\gamma$, Precisions $\bs{\nu}$, and thresholds $\bs{\tau}$
\State $i \gets 0$
\State $S \gets 0$ \label{line:param-s-init} \Comment{Track cumulative log-odds: $\sum_{j=0}^{i-1} \ln \left(\frac{\hat{p}_j}{1-\hat{p}_j}\right)$}
\State $P \gets 1$ \label{line:param-p-init} \Comment{Track for termination as Inequality~\eqref{eq:selection-termination}}
\Loop
    \State $p_i^* \gets \left(1 + \exp\left(\frac{\epsilon^\dagger}{\Delta} - S\right)\right)^{-1}$ \label{line:param-pi-star} \Comment{Adjusted target probability}
    \State $\nu_i \gets 1$
    \Loop \label{line:param-nu-loop} \Comment{Find minimal bit precision $\nu_i$}
        \State $\tau_i \gets \round{2^{\nu_i} p_i^*}$ \label{line:param-tau}
        \State $\hat{p}_i \gets \frac{\tau_i}{2^{\nu_i}}$ \label{line:param-phat}
        \State $a_i(\hat{p}_i) \gets -\ln \left(\frac{\hat{p}_i}{1-\hat{p}_i}\right) + S$
        \If{$\abs{a_i(\hat{p}_i) - \frac{\epsilon^\dagger}{\Delta}} \le \frac{\xi}{\Delta}$} \label{line:param-check}
            \State \textbf{break} \Comment{Minimal $\nu_i$ found}
        \EndIf
        \State $\nu_i \gets \nu_i + 1$
    \EndLoop
    
    \State Record $\nu_i$ into $\bs{\nu}$ and $\tau_i$ into $\bs{\tau}$ \label{line:param-record}
    \State $S \gets S + \ln \left(\frac{\hat{p}_i}{1-\hat{p}_i}\right)$ \label{line:param-s-update}
    
    \If{$i \geq \ceil{\log_2 \Delta}$}
        \State $P \gets P \cdot \hat{p}_i$ \label{line:param-p-update}
        \If{$P \leq \delta^\dagger$}
            \State $\gamma \gets i + 1$
            \State \textbf{break} \label{line:param-terminate} \Comment{Truncation error bound met}
        \EndIf
    \EndIf
    \State $i \gets i + 1$
\EndLoop
\State \Return $\gamma, \bs{\nu}, \bs{\tau}$
\end{algorithmic}}
\end{algorithm}

For Algorithm~\ref{alg:vddlm_circuit}, a naive, uniform parameter choice across all steps would either violate the strict differential privacy bounds because of accumulated numerical errors or result in excessively large bit-widths that degrade prover performance. Instead, we employ an iterative, dynamic parameter-selection strategy, formalized in Algorithm~\ref{alg:param_selection}. Specifically, given target privacy parameters $\epsilon^\dagger$ and $\delta^\dagger$, a sensitivity $\Delta$, and an error tolerance factor $\xi$ (where $0 < \xi < \epsilon^\dagger$), the algorithm dynamically selects the minimal required precision $\nu_i$ for each step $i$ so that the generated additive noise guarantees $\left(\hat{\epsilon}, \hat{\delta}\right)$-DP with $\hat{\epsilon} \in \left[\epsilon^\dagger - \xi, \epsilon^\dagger + \xi\right]$ and $\hat{\delta}\leq \delta^\dagger$.

To enforce the privacy bound, we track the multiplicative privacy loss in the additive log-odds space (initialized in Line~\ref{line:param-s-init}). At each step $i = 0, 1, 2, \dots$, rather than independently quantizing the ideal base probability $p_i^\dagger$, we compute an adjusted target probability $p_i^*$ (Line~\ref{line:param-pi-star}). Crucially, this computation is not static; it explicitly compensates for the exact numerical quantization errors that have accumulated over all preceding steps $j < i$ (maintained by the running sum $S$, which is updated in Line~\ref{line:param-s-update}):
\begin{equation}
    p_i^* = \left(1 + \exp\left(\frac{\epsilon^\dagger}{\Delta} - \sum_{j=0}^{i-1} \ln \frac{\hat{p}_j}{1-\hat{p}_j}\right)\right)^{-1}.
\end{equation} In particular, given the realized probabilities $\hat{p}_j$ for all $j < i$, we define the realized privacy-cost function at step $i$ as:
$a_i(p_i) = -\ln \left(\frac{p_i}{1-p_i}\right) + \sum_{j=0}^{i-1} \ln \left(\frac{\hat{p}_j}{1-\hat{p}_j}\right)$, where setting $p_i\gets p_i^*$ yields an output of $\frac{\epsilon^\dagger}{\Delta}$ exactly.

However, to ensure the corresponding Bernoulli variables $B_{+,i}$ and $B_{-,i}$ can be verifiably sampled efficiently, we further quantize the target probability $p_i^*$. In particular, starting at Line~\ref{line:param-nu-loop}, we iteratively find the minimal precision $\nu_i$ such that the absolute deviation in the privacy cost is below the error tolerance $\frac{\xi}{\Delta}$ (Line~\ref{line:param-check}). We denote the resulting integer threshold as $\tau_i = \round{2^{\nu_i} p_i^*}$ (Line~\ref{line:param-tau}), such that the fixed-point probability is $\hat{p}_i = \frac{\tau_i}{2^{\nu_i}}$ (Line~\ref{line:param-phat}).

Finally, the parameter-selection process halts upon reaching an acceptable level of failure probability $\hat{\delta}$ caused by the truncation error. The total number of bits per dimension, $\gamma$, is determined by tracking the cumulative tail probability $P$ (Line~\ref{line:param-p-update}) and stopping at the first index $i$ where the truncated probability mass falls below the target $\delta^\dagger$ (Line~\ref{line:param-terminate}):
\begin{equation}
    \prod_{j=\lceil \log_2 \Delta \rceil}^{\gamma-1} \hat{p}_j \le \delta^\dagger \label{eq:selection-termination}
\end{equation}

This dynamic scheme ensures that the implemented parameters $(\nu_i, \tau_i)$ yield a mechanism that satisfies $\left(\hat{\epsilon}, \hat{\delta}\right)$-DP with $\hat{\epsilon} \in \left[\epsilon^\dagger - \xi, \epsilon^\dagger + \xi\right]$ and $\hat{\delta}\leq \delta^\dagger$ (Section~\ref{sec:vddlm-analysis-privacy}) while actively optimizing the circuit complexity (Section~\ref{sec:vddlm-analysis-complexity}).

\subsection{Analysis of VDDLM}\label{sec:vddlm-analysis}

\subsubsection{Privacy} \label{sec:vddlm-analysis-privacy}

This subsection connects the sampling circuit construction to end-to-end privacy guarantees. At a high level, the analysis isolates two effects: a multiplicative privacy-loss term from local probability-ratio distortions introduced by quantization, and an additive term from the truncated tail mass. The theorem below formalizes this decomposition and quantifies both terms in the final $\left(\hat{\epsilon},\hat{\delta}\right)$ bound.

\begin{Thm}\label{thm:vddlm-dp}
    For any query $\mathcal{Q}: \mathcal{D} \to \Z^d$ with $L_1$ sensitivity $\Delta:=\sup_{\substack{D\simeq D'\\D, D'\in\D}}\norm{\mathcal{Q}\left(D\right) - \mathcal{Q}\left(D'\right)}_1$ that is a power of 2, the implemented mechanism is $\left(\hat{\epsilon}, \hat{\delta}\right)$-DP, where
    \begin{align}
        \hat{\epsilon}\left(\Delta\right) &:= \Delta \cdot \max_{i\in\range{\gamma}}\left(\ln\frac{1-\hat{p}_i}{\hat{p}_i}-\sum_{j\in\range{i}}\ln \frac{1-\hat{p}_j}{\hat{p}_j}\right) \in \left[\epsilon^\dagger - \xi, \epsilon^\dagger + \xi\right] \label{eq:vddlm-dp-eps} \\
        \hat{\delta}\left(\Delta\right) &:= \prod_{i=\log_2 \Delta}^{\gamma-1} \hat{p}_i \leq \delta^\dagger. \label{eq:vddlm-dp-delta}
    \end{align}
    These bounds are asymptotically tight.
\end{Thm}

The power-of-two sensitivity assumption can be removed by a more careful analysis of the tail region; we restate the theorem for general integer sensitivity in Appendix~\ref{app:vddlm-proofs}.

\begin{proof}[Proof sketch of Theorem~\ref{thm:vddlm-dp}]
    We compare the output distributions on neighbouring databases by aligning outcomes with the same shifted query value. The argument separates outcomes into two groups: a dominant region where probabilities can be matched up to a multiplicative factor, and a tail region near truncation boundaries that is charged to an additive slack term. The first part yields the $e^{\hat{\epsilon}(\Delta)}$ term, and the second part yields $\hat{\delta}(\Delta)$. Combining them gives the claimed $\left(\hat{\epsilon},\hat{\delta}\right)$ guarantee. Full details are provided in Appendix~\ref{app:vddlm-proofs}.
\end{proof}

\subsubsection{Circuit Complexity} \label{sec:vddlm-analysis-complexity}
The complexity analysis quantifies the bit cost of implementing the sampler in zero-knowledge circuits. The main driver is the total number of Bernoulli-source bits per dimension, $\sum_{i\in\range{\gamma}}\nu_i$, where $\nu_i$ is the precision at level $i$ and $\gamma$ is the truncation depth. The parameter-selection rule balances these two quantities: larger precisions reduce per-step approximation error, while larger depth reduces truncation error. The theorem below captures the resulting asymptotic trade-off under the target privacy parameters.

\begin{Thm}\label{thm:vddlm-complexity}
    The complexity of the sampling circuit, measured by the total number of bits required per dimension, is bounded by $\bigO{\log\frac{1}{\hat\delta} + \log \frac{\Delta}{\xi}\left(\log \frac{\Delta}{\hat\epsilon} + \log\log\frac{1}{\hat\delta}\right)}$.
\end{Thm}

\begin{proof}[Proof sketch of Theorem~\ref{thm:vddlm-complexity}]
    The proof has two steps. First, we upper-bound the per-level precision needed to keep approximation error under control at each level, which gives a bound on each $\nu_i$. Second, we upper-bound how many levels are needed before the remaining tail probability falls below the target failure probability, which gives a bound on $\gamma$. Summing the per-level bit costs over all levels then yields the total complexity bound. Finally, we translate from design-time targets to realized privacy parameters using the parameter-selection guarantees. Full details are in Appendix~\ref{app:vddlm-proofs}.
\end{proof}

The following proposition instantiates Theorem~\ref{thm:vddlm-complexity} under representative settings of the relative error tolerance $\xi$ relative to $\epsilon^\dagger$, highlighting the resulting asymptotic regimes.

\begin{Prop}\label{prop:vddlm-complexity-superior}
    By setting $\xi = 0.05\epsilon^\dagger$ (or, more generally, any constant fraction of $\epsilon^\dagger$), the asymptotic complexity becomes $\bigO{\log \frac{1}{\hat\delta} + \log^2 \frac{\Delta}{\hat\epsilon} + \log \frac{\Delta}{\hat\epsilon} \log\log\frac{1}{\hat{\delta}}}$.  Similarly, in the cryptographic setting where $\hat{\delta} = 2^{-\kappa}$, this bound becomes
        $\bigO{\kappa + \log\frac{\Delta}{\hat{\epsilon}}\left(\log\frac{\Delta}{\hat{\epsilon}} + \log \kappa\right)}$
    when $\xi$ is a constant fraction of $\epsilon^\dagger$, or alternatively, $\bigO{\kappa + \left(\log\frac{\Delta}{\hat{\epsilon}} + \log \kappa\right)^2}$ when $\xi$ is inverse polynomially small in $\kappa$.
\end{Prop}

The circuit complexity bounds in Proposition~\ref{prop:vddlm-complexity-superior} are asymptotically superior to prior constructions. Methods relying on the statistical distance between the ideal and implemented distributions typically require $\bigOmega{\log \frac{1}{\hat\delta} \left(\log \frac{\Delta}{\hat\epsilon} + \log \log \frac{1}{\hat\delta}\right)}$, or equivalently, $\bigOmega{\kappa \left(\log \frac{\Delta}{\hat\epsilon} + \log \kappa\right)}$~\cite{DBLP:conf/ccs/WeiYFCW23}. Furthermore, other secure sampling circuits designed for secure and verifiable computations of DP mechanisms exhibit a complexity bounded by $\bigOmega{\frac{\Delta}{\hat{\epsilon}}\log\frac{1}{\hat{\delta}}}$, or equivalently, $\bigOmega{\frac{\Delta}{\hat{\epsilon}}\cdot\kappa}$~\cite{DBLP:conf/ccs/FranzeseF00JPD25,BC23}.

\section{Experiments on VDDLM}\label{sec:experiments}

\begin{figure*}[!htbp]
    \includegraphics[width=\linewidth]{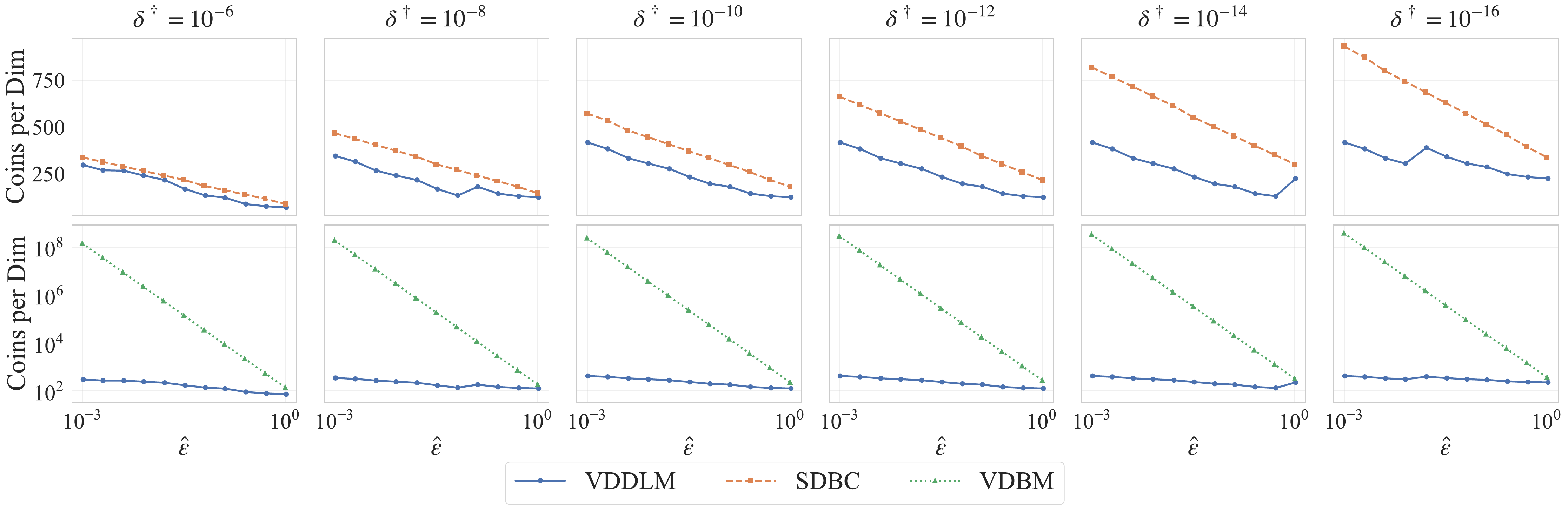}
    \caption{Circuit complexity of the noise sampling circuit, measured by the number of fair coin flips per output dimension, for VDDLM, the statistical-distance-based construction (SDBC)~\cite{DBLP:conf/ccs/WeiYFCW23}, and VDBM~\cite{BC23}.}
    \label{fig:complexity}
\end{figure*}

\begin{figure}[!htbp]
    \includegraphics[width=\linewidth]{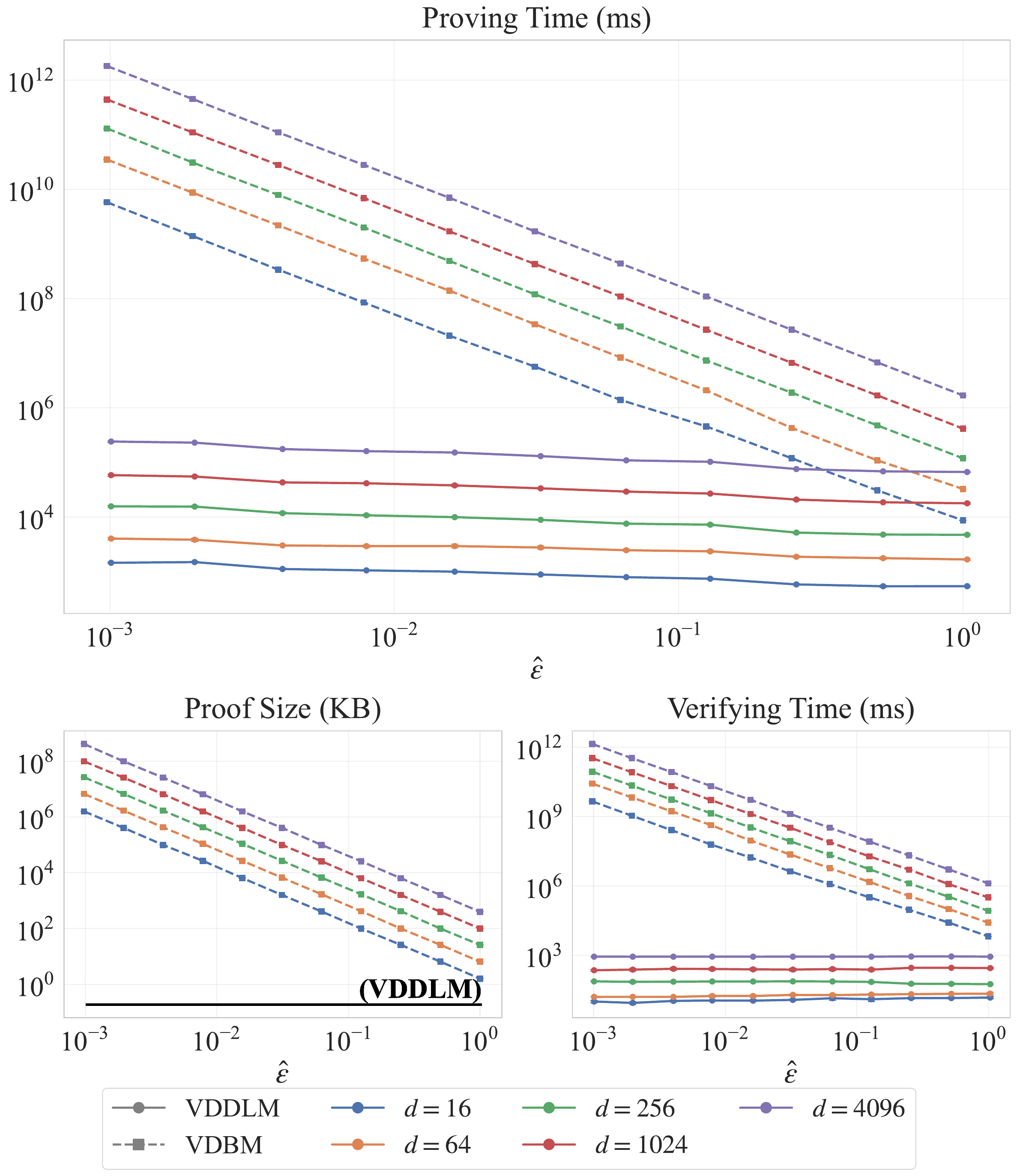}
    \caption{Per-server overhead of VDDLM versus VDBM in $\Pi_\ser$. VDDLM achieves a proving-time speedup from $50\times$ to $4{,}000{,}000\times$ over VDBM, with the advantage growing as $\epsilon$ decreases. Running times above $10^7$ ms are estimated.}
    \label{fig:runtime_vs_vdbm}
\end{figure}

\begin{figure*}[!htbp]
    \includegraphics[width=\linewidth]{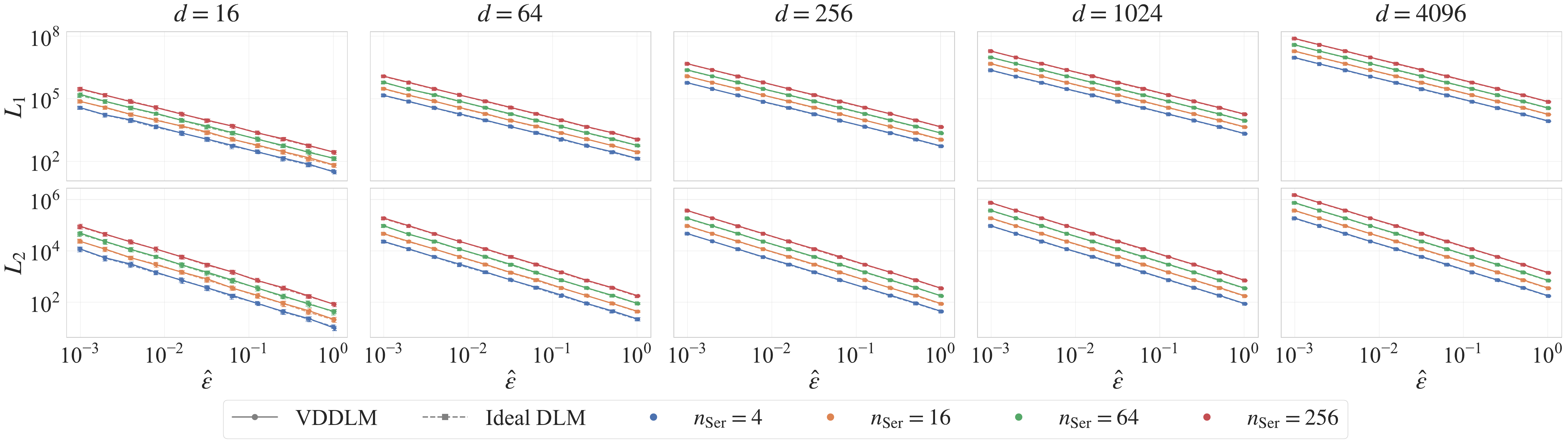}
    \caption{Privacy--utility trade-off of VDDLM, measured by the $L_1$ and $L_2$ errors of the noisy output aggregated from $n_\ser$ servers relative to the ideal DLM.}
    \label{fig:error_analysis}
\end{figure*}

We implement VDDLM in Rust using the \texttt{arkworks} library~\cite{arkworks}. We use the BLS12-381 curve, one of the most prevalent elliptic curves in modern ZKP systems, which provides 128-bit security~\cite{bls}. The implementation and artifacts are anonymously available at \VDDPCodeURL. We ran the experiments with 16 cores allocated from an Intel Xeon Platinum 8358 CPU @ 2.60GHz and 64 GB of memory. The comparisons focus on $\Pi_\ser$, the major difference between VDDLM and VDBM, because the instantiation of $\Pi_\cli$ is application-specific, and any valid $\Pi_\cli$ for VDDLM is also valid for VDBM under the same setting. Unless otherwise specified, we report results with allowable privacy-budget deviation controlled as $\xi = 0.05\epsilon^\dagger$ and allowable failure probability $\delta^\dagger = 10^{-10}$. We primarily evaluate VDDLM against VDBM~\cite{BC23}, which operates under the same client-server-verifier setting.

For all baselines (e.g., VDBM and the alternative circuit construction via statistical distance), parameter configurations are computed directly from the target privacy budget $\left(\epsilon^\dagger, \delta^\dagger\right)$. However, unlike VDDLM, these baselines do not support more precise post hoc privacy accounting analogous to Theorem~\ref{thm:vddlm-dp}. Therefore, we treat $\hat{\epsilon} = \epsilon^\dagger$ and $\hat{\delta} = \delta^\dagger = 10^{-10}$ as their actual (best-known) privacy budgets. Meanwhile, for VDDLM, we compute the realized privacy budget via Theorem~\ref{thm:vddlm-dp}. Therefore, under the same $\left(\epsilon^\dagger, \delta^\dagger\right)$, the realized $\hat{\epsilon}$ may differ slightly from the baselines, with a controlled misalignment of $\abs{\epsilon^\dagger - \hat{\epsilon}} \leq 0.05\epsilon^\dagger$. Additionally, by executing Algorithm~\ref{alg:param_selection}, the failure probability in VDDLM's privacy guarantee is upper-bounded by $\hat{\delta}\leq 10^{-10}$.

First, to validate Algorithms~\ref{alg:vddlm_circuit}~and~\ref{alg:param_selection}, we measure the complexity of the noise sampling circuit, represented by the number of fair coin flips required as circuit input. We compare our construction against two baselines: the statistical-distance-based construction~\cite{DBLP:conf/ccs/WeiYFCW23} and the binomial-distribution-based construction used in VDBM~\cite{BC23}, and report the results in Figure~\ref{fig:complexity}. Consistent with the asymptotic analysis in Section~\ref{sec:vddlm-analysis-complexity}, the complexity advantage of VDDLM over both baselines grows as the privacy budget $\epsilon$ or failure probability $\delta$ decreases. Furthermore, our construction achieves a smaller multiplicative complexity than the statistical-distance-based construction, which supports its applicability to cryptography-based secure and verifiable computations of the Discrete Laplace Mechanism.

Next, we evaluate the efficiency and scalability of VDDLM by reporting the per-server overhead of $\Pi_\ser$ in Figure~\ref{fig:runtime_vs_vdbm}. VDDLM achieves a proving-time speedup ranging from $50\times$ at larger privacy budgets to approximately $4{,}000{,}000\times$ at smaller privacy budgets, demonstrating that its scalability advantage becomes particularly pronounced in the small-$\epsilon$ regime. This overhead reduction is primarily attributable to the logarithmic dependence of circuit complexity on $\frac{\Delta}{\hat{\epsilon}}$, in contrast to the linear dependence for VDBM. The accompanying reductions in proof size and verification time also partially reflect the choice of ZKP backend, which we do not claim as a primary contribution of this work: because of the succinct proof structure of the chosen backend for VDDLM, its proof size remains constant with respect to circuit complexity; similarly, as verification time scales sublinearly with respect to circuit complexity for VDDLM and linearly for VDBM, a $10^3\times$--$10^9\times$ speedup is observed.

Finally, to validate the privacy--utility trade-off of the constructed sampling circuit, we measure the $L_1$ and $L_2$ errors of the noisy output aggregated from $n_\ser$ servers and compare VDDLM against the ideal DLM. The results in Figure~\ref{fig:error_analysis} show that the error curves of the ideal DLM and VDDLM nearly coincide across all evaluated settings, confirming that the fixed-point approximations and truncation introduced by the sampling circuit do not significantly affect the privacy--utility trade-off.
\section{Limitations and Discussion}

We acknowledge several limitations of the current framework and instantiation. Our main analysis focuses on linear queries with per-server noise sampled independently and verified through fixed circuits. Extending VDDP to mechanisms that require multi-round, data-dependent communication among servers (e.g., MPC-DP realizations of the Sparse Vector Technique (SVT) or exponential mechanisms~\cite{DBLP:conf/uss/BohlerK20,DBLP:conf/ccs/BohlerK21}) remains an important direction. However, this setting raises additional questions about controlling the overhead induced by cryptographic primitives (e.g., ZKPs superimposed on MPC) to a tolerable level in realistic deployments.

We also identify two open directions for improving the efficiency and generality of VDDP. First, it is natural to ask whether non-zero-knowledge proof systems can instantiate VDDP while asymptotically reducing computation or communication overhead and still achieving provably bounded privacy leakage against potentially colluding adversaries. Second, extending VDDP to Discrete Gaussian Mechanisms~\cite{dgauss,ddgm} remains challenging. Existing rejection-sampling-based constructions induce variable-depth computations that are not naturally compatible with fixed-circuit ZKP backends; workarounds that expose acceptance bits can increase verifier overhead and complicate tight privacy accounting under fixed-point arithmetic. Progress in this direction will likely require new sampling techniques that avoid rejection loops, together with refined analyses that explicitly capture numerical imprecision induced by circuit-compatible implementations.
\section{Conclusions}

In this study, we rigorously defined verifiable distributed differential privacy and systematically explored its relationship with zero-knowledge. We further demonstrated the feasibility of VDDP through a concrete instantiation, VDDLM, which improves the efficiency--privacy--utility trade-off compared with previous state-of-the-art constructions. Our experiments show substantial practical gains, including proving-time speedups from $50\times$ to approximately $4{,}000{,}000\times$ over VDBM under the same client-server-verifier setting. Overall, we view VDDP as a foundation for verifiable and private data analysis in distributed systems, and VDDLM as an efficient point in this design space. We expect future work on richer mechanisms, proof systems, and numerical analysis to further narrow the gap between formal privacy guarantees and practical deployability.

\bibliographystyle{ACM-Reference-Format}
\bibliography{reference}
\appendix
\section{Non-ZK Variant for VDBM}\label{app:nonzk-vdbm-variant}

In this appendix, we present a modification of $\Pi_\sf{Bin}$, the protocol used by VDBM, that violates the zero-knowledge property yet preserves DP. First, the original version of $\Pi_\sf{Bin}$ is perfectly zero-knowledge; that is, the distributions of proof transcripts for $x = 0$ and $x = 1$ are statistically identical. By considering $0$ and $1$ as neighbouring inputs, this is equivalent to $0$-DP. Therefore, we generalize the original version to privacy parameters $\epsilon' > 0$.

Let $g$ and $h$ be public generators of the commitment group, and let $c_x = g^x h^{r_x}$ denote the commitment to the bit $x$ with hiding factor $r_x$.

\begin{figure}[H]
    \centering
    \autofit{\begin{pcvstack}[noindent]
    \pseudocode{
        \prover_\sf{Bin}\left(x=0, r_x, c_x; g, h\right)  \< \< \verifier_\sf{Bin}\left(c_x; g, h\right)\\[0.1 \baselineskip][\hline]
        \pcln b, e_1 \sample \F, v_1 \sample \mathcal{P}_{K, \xi} \< \< \\
        \pcln d_0 \gets h^b, d_1 \gets h^{v_1}\left(\frac{c_x}{g}\right)^{-e_1}  \< \sendmessageright*[2cm]{(d_0, d_1)} \< (d_0, d_1) \\
        \pcln e_0 \gets e - e_1, v_0 = b + e_0 r_x \< \sendmessageleft*[2cm]{e} \< e\sample \F\\
        \pcln \< \sendmessageright*[2cm]{(v_0, e_0, v_1, e_1)} \< e_0 + e_1 \checkeq e \land c_x^{e_0} \checkeq h^{v_0} \land c_x^{e_1} \checkeq g^{e_1}h^{v_1}
    }
    \pseudocode{
        \prover_\sf{Bin}\left(x=1, r_x, c_x; g, h\right)  \< \< \verifier_\sf{Bin}\left(c_x; g, h\right)\\[0.1 \baselineskip][\hline]
        \pcln b, e_0 \sample \F, v_0 \sample \mathcal{P}_{K, \xi} \< \< \\
        \pcln d_0 \gets h^{v_0}c_x^{-e_0}, d_1 \gets h^b \< \sendmessageright*[2cm]{(d_0, d_1)} \< (d_0, d_1) \\
        \pcln e_1 \gets e - e_0, v_1 = b + e_1 r_x \< \sendmessageleft*[2cm]{e} \< e\sample \F\\
        \pcln \<\sendmessageright*[2cm]{(v_0, e_0, v_1, e_1)} \< e_0 + e_1 \checkeq e \land c_x^{e_0} \checkeq h^{v_0} \land c_x^{e_1} \checkeq g^{e_1}h^{v_1}
    }
    \end{pcvstack}}

    \caption{$\Pi_{\sf{Bin}}^{K, \xi}$, protocol for proving the knowledge of $x \in \bin$ and $r_x \in \F$ such that $c_x = g^x h^{r_x}$. $\xi = 1$ corresponds to the original zero-knowledge version~\cite{DBLP:conf/crypto/CramerDS94, DBLP:conf/eurocrypt/Damgard00, DBLP:journals/ftsec/Thaler22}.}
    \label{fig:bin-modified}
\end{figure}

We present the modified protocol $\Pi_{\sf{Bin}}^{K, \xi}$ in Figure~\ref{fig:bin-modified}, where $\mathcal{P}_{K, \xi}$ is defined for any integer $1 \leq K \leq \frac{\abs{\F} - 1}{2}$ and real number $\xi \geq 1$ such that
\begin{equation}
    \Pr\left[k \sample \mathcal{P}_{K, \xi}\right] = \case{\frac{2}{\xi + 1}\frac{1}{\abs{\F}} & \text{if } k = 0, 2, 4, \dots, 2K - 2 \\ \frac{2\xi}{\xi + 1}\frac{1}{\abs{\F}} & \text{if } k = 1, 3, 5, \dots, 2K - 1 \\ \frac{1}{\abs{\F}} & \text{otherwise}}.
\end{equation}
The choice $\xi = 1$ corresponds to the original zero-knowledge version; setting $\xi > 1$ is the only modification made to the original $\Pi_\sf{Bin}$.

\begin{Lem}\label{lem:bin-modified-dp}
    There exists a simulator $\simulator_{\sf{Bin}}^{K, \xi}\left(x, c_x; g, h\right)$ that perfectly simulates the proof transcript of $\prover_\sf{Bin}\left(x, r_x, c_x; g, h\right)\leftrightarrow\verifier_\sf{Bin}\left(c_x; g, h\right)$, such that $\simulator_{\sf{Bin}}^{K, \xi}\left(\cdot, c_x; g, h\right)$ is $\left(\log \xi, 0\right)$-DP.
\end{Lem}

\begin{proof}[Proof of Lemma~\ref{lem:bin-modified-dp}]
    $\simulator_{\sf{Bin}}^{K, \xi}$ can be constructed by sampling $\left(v_0, v_1\right)$ from $\F \times \mathcal{P}_{K, \xi}$ when $x = 0$ and from $\mathcal{P}_{K, \xi} \times \F$ when $x = 1$, as specified below.

    \begin{figure}[H]
        \centering
        \procedureblock[space=auto, linenumbering]{$\simulator_{\sf{Bin}}^{K, \xi}\left(x, c_x; g, h\right)$}{
            \pcif x = 1 \pcdo\\
                (v_0, e_0, v_1, e_1) \sample \mathcal{P}_{K, \xi} \times \F \times \F \times \F\\
            \pcelseif x = 0 \pcdo\\
                (v_0, e_0, v_1, e_1) \sample \F \times \F \times \mathcal{P}_{K, \xi} \times \F\\
            \pcendif \\
            e \gets e_0 + e_1\\
            d_0 \gets c_x^{-e_0} \cdot h^{v_0}\\
            d_1 \gets c_x^{-e_1} \cdot g^{e_1} h^{v_1}\\
            \pcreturn{v_0, e_0, v_1, e_1, e, d_0, d_1}
        }
        \caption{Construction of $\simulator_{\sf{Bin}}^{K, \xi}$.}
        \label{fig:bin-modified-dp-sim}
    \end{figure}

    Let $\mathcal{G}(x)$ denote the distribution of $\left(v_0,v_1\right)$ simulated for input $x$. For any $\left(a, b\right) \in \F^2$,
    \begin{equation}
        \frac{\Pr[\mathcal{G}(0) = (a, b)]}{\Pr[\mathcal{G}(1) = (a, b)]} \leq \frac{\frac{1}{\abs{\F}} \cdot \frac{2\xi}{\xi + 1} \frac{1}{\abs{\F}}}{\frac{2}{\xi + 1} \frac{1}{\abs{\F}} \cdot \frac{1}{\abs{\F}}} = \xi.
    \end{equation}
    Symmetrically, it holds that $\frac{\Pr[\mathcal{G}(1) = (a, b)]}{\Pr[\mathcal{G}(0) = (a, b)]} \leq \xi$. Therefore, $\mathcal{G}$ is $\left(\log \xi, 0\right)$-DP. Consequently, the sampling of $\left(v_0, v_1\right)$ is $\left(\log \xi, 0\right)$-DP, and by post-processing, $\simulator_{\sf{Bin}}^{K, \xi}\left(\cdot, c_x; g, h\right)$ is $\left(\log \xi, 0\right)$-DP.
\end{proof}

\section{Additional Details of VDDLM}\label{app:vddlm-proofs}

\providecommand{\p}{\mathfrak{p}}
\providecommand{\pof}[1]{\mathfrak{p}\left(#1\right)}
\providecommand{\q}{\mathfrak{q}}
\providecommand{\qof}[1]{\mathfrak{q}\left(#1\right)}

In this appendix, we collect the technical details of VDDLM deferred from Section~\ref{sec:vddlm-analysis}. We first present auxiliary lemmas on the one-sided and two-sided distributions used by the privacy proof, then restate the privacy theorem without the power-of-two sensitivity assumption, and finally provide the tightness and complexity details.

Here, $\q$ denotes the probability mass function of each one-sided component ($X_+$ and $X_-$), and $\p$ denotes the probability mass function of their difference $X_+ - X_-$.

\begin{Lem}\label{lem:q-decreasing}
    For any integer $s$ such that $0 \leq s \leq 2^\gamma - 2$, the following inequality holds:
    \begin{equation}
        1 < e^{\frac{\epsilon^\dagger - \xi}{\Delta}} \leq \frac{\qof{s}}{\qof{s+1}} \leq e^{\hat{\epsilon}\left(1\right)} \leq e^{\frac{\epsilon^\dagger + \xi}{\Delta}}.
    \end{equation}
\end{Lem}

\begin{proof}[Proof of Lemma~\ref{lem:q-decreasing}]
    Represent $s$ and $s+1$ in binary as
    \begin{align}
        s &= \sum_{i=0}^{\gamma-1}2^i b_i,\\
        s+1 &= \sum_{i=0}^{\gamma-1}2^i b_i',
    \end{align}
    Comparing the binary representations, there exists a unique index $0\leq i^* \leq \gamma - 1$ such that $b_i = b_i'$ for $i > i^*$, $b_{i^*}' = 1$ and $b_{i^*} = 0$, and $b_i' = 0$ and $b_i = 1$ for $i < i^*$. Therefore,
    \begin{equation}
        \frac{\qof{s}}{\qof{s+1}} = \frac{1-\hat{p}_{i^*}}{\hat{p}_{i^*}}\cdot \prod_{i=0}^{i^*-1}\frac{\hat{p}_i}{1-\hat{p}_i},
    \end{equation} where, by the execution of Algorithm~\ref{alg:param_selection} and the definition of $\hat{\epsilon}\left(\cdot\right)$ (Equation~\eqref{eq:vddlm-dp-eps}, Theorem~\ref{thm:vddlm-dp}),
    \begin{equation}
        1 < e^{\frac{\epsilon^\dagger - \xi}{\Delta}} \leq \frac{\qof{s}}{\qof{s+1}} \leq e^{\hat{\epsilon}\left(1\right)} \leq e^{\frac{\epsilon^\dagger + \xi}{\Delta}}.
    \end{equation}
\end{proof}

\begin{Lem}\label{lem:dlm-prob}
    For any two integers $s, t$ such that $0 \leq \abs{s} \leq \abs{t} \leq 2^\gamma - 1$, the following inequality holds:
    \begin{equation}
        \pof{t} \leq \pof{s} \leq e^{\hat{\epsilon}\left(\abs{t} - \abs{s}\right)}\pof{t} + \sum_{i=2^\gamma - \abs{t}}^{2^\gamma - \abs{s} - 1}\qof{i}\qof{i+\abs{s}}. \label{eq:double-sided-pmf-bound}
    \end{equation}
\end{Lem}

\begin{proof}[Proof of Lemma~\ref{lem:dlm-prob}]
    Without loss of generality, consider $0 \leq s \leq t \leq 2^\gamma - 1$. We have:
    \begin{align}
        \pof{s} =& \sum_{i\in\range{2^\gamma - s}}\qof{i}\qof{i+s}\\
        =& \sum_{i\in\range{2^\gamma - t}}\qof{i}\qof{i+s} + \sum_{i=2^\gamma-t}^{2^\gamma-s-1}\qof{i}\qof{i+s}.
    \end{align}
    Furthermore, by repeated application of Lemma~\ref{lem:q-decreasing}, for any $j' \geq i' \geq 0$, $\qof{j'} \leq \qof{i'}\leq e^{\hat{\epsilon}\left(j'-i'\right)}\qof{j'}$. Substituting $i'\gets i+s$ and $j'\gets i+t$ gives Inequality~\eqref{eq:double-sided-pmf-bound}.
\end{proof}

Furthermore, we restate Theorem~\ref{thm:vddlm-dp} without the power-of-two sensitivity assumption and provide the full proof.

\begin{Thm}[Theorem~\ref{thm:vddlm-dp}, Restatement]\label{thm:vddlm-dp-general}
    For any query $\mathcal{Q}: \mathcal{D} \to \Z^d$ with $L_1$ sensitivity $\Delta:=\sup_{\substack{D\simeq D'\\D, D'\in\D}}\norm{\mathcal{Q}\left(D\right) - \mathcal{Q}\left(D'\right)}_1$, the implemented mechanism is $\left(\hat{\epsilon}, \hat{\delta}\right)$-DP, where
    \begin{align}
        \hat{\epsilon}\left(\Delta\right) &:= \Delta \cdot \max_{i\in\range{\gamma}}\left(\ln\frac{1-\hat{p}_i}{\hat{p}_i}-\sum_{j\in\range{i}}\ln \frac{1-\hat{p}_j}{\hat{p}_j}\right) \in \left[\epsilon^\dagger - \xi, \epsilon^\dagger + \xi\right]\\
        \hat{\delta}\left(\Delta\right) &:= \sum_{i=2^\gamma - \Delta}^{2^\gamma - 1}\qof{i} \leq \delta^\dagger .
    \end{align}
    Furthermore, these bounds are asymptotically tight (Lemma~\ref{lem:tightness}).
\end{Thm}

\begin{proof}[Proof of Theorem~\ref{thm:vddlm-dp-general}]
    To prove Theorem~\ref{thm:vddlm-dp-general}, it is sufficient to prove that for any pair of neighbouring databases $D\simeq D'$ and any $\Delta\in\N$ such that $\norm{\mathcal{Q}\left(D\right) - \mathcal{Q}\left(D'\right)}_1 = \Delta$, and any subset $T \subseteq \Z^d$, we have \begin{equation}
        \prob{\Mof{D} \in T} \leq e^{\hat{\epsilon}\left(\Delta\right)}\prob{\Mof{D'} \in T} + \hat{\delta}\left(\Delta\right).\label{eq:main-ineq}
    \end{equation}

    \newcommand{\Mres}{\left.\M\middle|_{\left\{D, D'\right\}}\right.}
    \newcommand{\mechres}[1]{\left.#1\middle|_{\left\{D, D'\right\}}\right.}

    We further argue it is sufficient to prove Inequality~\eqref{eq:main-ineq} for single-dimensional queries (i.e., $d=1$). To see this, assume Inequality~\eqref{eq:main-ineq} holds for single-dimensional queries, and consider a multidimensional query \begin{equation}
        \M: D \mapsto \left(\M^{(0)}\left(D\right), \M^{(1)}\left(D\right), \dots, \M^{(d-1)}\left(D\right)\right),
    \end{equation} such that for any $D\in \D$, the components $\M^{(i)}\left(D\right):\D\to\Z\,\left(i\in \range{d}\right)$ are mutually independent. For each $i\in\range{d}$, let $\mathcal{Q}^{\left(i\right)}\left(\cdot\right)$ denote the $i$th output coordinate of the unperturbed query $\mathcal{Q}$. Then, for any pair of neighbouring databases $D\simeq D'$, consider the restriction of $\M$ to $\left\{D, D'\right\}$, i.e., $\Mres\equiv\left(\mechres{\M^{(i)}}\right)_{i\in\range{d}}$. Let $\Delta^{(i)} = \abs{\mathcal{Q}^{\left(i\right)}\left(D\right) - \mathcal{Q}^{\left(i\right)}\left(D'\right)}$ denote the sensitivity of $\mechres{\M^{(i)}}$. This restricted component is therefore $\left(\hat{\epsilon}\left(\Delta^{(i)}\right), \hat{\delta}\left(\Delta^{(i)}\right)\right)$-DP. By mutual independence and the law of sequential composition~\cite{DR14}, the privacy budget of $\Mres$ is \begin{equation}
        \left(\sum_{i\in\range{d}}\hat{\epsilon}\left(\Delta^{(i)}\right), \sum_{i\in\range{d}}\hat{\delta}\left(\Delta^{(i)}\right)\right),
    \end{equation} where the first term is exactly $\hat{\epsilon}\left(\Delta\right)$, and by the monotonicity of $\q$ (Lemma~\ref{lem:q-decreasing}), the second term is upper-bounded by $\hat{\delta}\left(\Delta\right)$. Therefore, Inequality~\eqref{eq:main-ineq} holds for any pair of neighbouring databases $D\simeq D'$ such that $\norm{\mathcal{Q}\left(D\right) - \mathcal{Q}\left(D'\right)}_1 = \Delta$.
    
    For the remainder of this proof, we establish Inequality~\eqref{eq:main-ineq} with $d=1$. Without loss of generality, consider a pair of neighbouring databases $D\simeq D'$ such that $\mathcal{Q}\left(D'\right) - \mathcal{Q}\left(D\right) = \Delta$, where $\mathcal{Q}$ is the query under analysis. For any subset of integers $S\subset \left[-2^\gamma + 1, 2^\gamma-1\right]$, we observe:
    \begin{align}
        & \prob{\Mof{D} \in \mathcal{Q}\left(D\right) + S} \\
        =&\sum_{\substack{s\in S\\-2^\gamma+1\leq s \leq -2^\gamma+\Delta}}\pof{s} + \sum_{\substack{s\in S\\-2^\gamma+\Delta+1\leq s \leq 2^\gamma-1}}\pof{s}.
    \end{align}
    We can bound the first term as follows:
    \begin{align}
        ~&\sum_{\substack{s\in S\\-2^\gamma+1\leq s \leq -2^\gamma+\Delta}}\pof{s}\leq\sum_{i=-2^\gamma+1}^{-2^\gamma+\Delta}\pof{i}=\sum_{\substack{i, j\in \range{2^\gamma}\\j-i\geq 2^\gamma - \Delta}}\qof{i}\qof{j}.
    \end{align}
    On the other hand, for any $-2^\gamma + \Delta + 1\leq s \leq 2^\gamma - 1$, \begin{equation}
        \prob{\Mof{D'} = \mathcal{Q}\left(D\right) + s} = \pof{s - \Delta}.
    \end{equation} Therefore, for any $-2^\gamma + \Delta + 1\leq s \leq \floor{\frac{\Delta}{2}}$, since $\abs{s} \leq \abs{s-\Delta}$, Lemma~\ref{lem:dlm-prob} yields:
    \begin{align}
        \pof{s}\leq & e^{\hat{\epsilon}\left(\abs{s-\Delta} - \abs{s}\right)}\pof{s-\Delta} + \sum_{i=2^\gamma - \abs{s-\Delta}}^{2^\gamma - \abs{s} - 1}\qof{i}\qof{i+\abs{s}}\\
        \leq & e^{\hat{\epsilon}\left(\Delta\right)}\pof{s-\Delta} + \sum_{i=2^\gamma - \abs{s-\Delta}}^{2^\gamma - \abs{s} - 1}\qof{i}\qof{i+\abs{s}}.
    \end{align}
    Specifically, for $-2^\gamma + \Delta + 1\leq s \leq 0$:
    \begin{equation}
        \pof{s}\leq e^{\hat{\epsilon}\left(\Delta\right)}\pof{s-\Delta} + \sum_{\substack{i, j\in \range{2^\gamma}\\j\geq 2^\gamma - \Delta\\j-i=-s}}\qof{i}\qof{j}.
    \end{equation}
    For $1\leq s \leq \floor{\frac{\Delta}{2}}$:
    \begin{align}
        \pof{s}\leq& e^{\hat{\epsilon}\left(\Delta\right)}\pof{s-\Delta} + \sum_{\substack{i, j\in \range{2^\gamma}\\j\geq 2^\gamma + s - \Delta\\i-j=s}}\qof{i}\qof{j}\\
        \leq& e^{\hat{\epsilon}\left(\Delta\right)}\pof{s-\Delta} + \sum_{\substack{i, j\in \range{2^\gamma}\\j\geq 2^\gamma - \Delta\\j-i=-s}}\qof{i}\qof{j}.
    \end{align}
    The condition $j\geq 2^\gamma - \Delta$ follows from reindexing the tail term in Lemma~\ref{lem:dlm-prob}. For $s\leq 0$, setting $j=i+\abs{s}=i-s$ transforms the lower bound $i\geq 2^\gamma-\abs{s-\Delta}=2^\gamma-\Delta+s$ into $j\geq 2^\gamma-\Delta$. For $s>0$, reindexing gives the stronger condition $j\geq 2^\gamma+s-\Delta$, which we relax to $j\geq 2^\gamma-\Delta$ to obtain a common upper bound for both cases.
    Additionally, for $s > \floor{\frac{\Delta}{2}}$, as $\abs{s} > \abs{s - \Delta}$,
    \begin{equation}
        \pof{s} \leq \pof{s - \Delta} < e^{\hat{\epsilon}\left(\Delta\right)}\pof{s - \Delta}.
    \end{equation}

    Therefore, combining these bounds gives
    \begin{align}
        ~ & \prob{\Mof{D} \in \mathcal{Q}\left(D\right) + S} - e^{\hat{\epsilon}\left(\Delta\right)}\prob{\Mof{D'} \in \mathcal{Q}\left(D\right) + S}\\
        \leq & \sum_{\substack{i, j\in \range{2^\gamma}\\j-i\geq 2^\gamma - \Delta}}\qof{i}\qof{j} + \sum_{\substack{-2^\gamma + \Delta + 1\leq s\leq 2^\gamma - 1\\s\in S}}\left(\pof{s} - e^{\hat{\epsilon}\left(\Delta\right)}\pof{s-\Delta}\right)\\
        \leq & \sum_{\substack{i, j\in \range{2^\gamma}\\j-i\geq 2^\gamma - \Delta}}\qof{i}\qof{j} + \sum_{\substack{-2^\gamma + \Delta + 1\leq s\leq \floor{\frac{\Delta}{2}}\\s\in S}} \sum_{\substack{i, j\in \range{2^\gamma}\\j\geq 2^\gamma - \Delta\\j-i=-s}}\qof{i}\qof{j}\\
        \leq & \sum_{\substack{i, j\in \range{2^\gamma}\\j-i\geq 2^\gamma - \Delta}}\qof{i}\qof{j} + \sum_{s = -2^\gamma + \Delta + 1}^{\floor{\frac{\Delta}{2}}} \sum_{\substack{i, j\in \range{2^\gamma}\\j\geq 2^\gamma - \Delta\\j-i=-s}}\qof{i}\qof{j}\\
        = & \sum_{\substack{i, j\in \range{2^\gamma}\\j\geq 2^\gamma - \Delta\\j-i\geq 2^\gamma - \Delta}}\qof{i}\qof{j} + \sum_{\substack{i, j\in \range{2^\gamma}\\j\geq 2^\gamma - \Delta\\-\floor{\frac{\Delta}{2}}\leq j-i\leq 2^\gamma - \Delta - 1}}\qof{i}\qof{j}\\
        \leq & \sum_{\substack{i, j\in \range{2^\gamma}\\j\geq 2^\gamma - \Delta}} \qof{i}\qof{j} \\
        = & \sum_{j = 2^\gamma - \Delta}^{2^\gamma - 1} \qof{j} \\
        =: & \hat{\delta}\left(\Delta\right),
    \end{align} where the equality rewrites $j-i=-s$ over the stated range of $s$. The two pair sets are disjoint, and both are contained in the set of pairs satisfying $j\geq 2^\gamma - \Delta$.

    Additionally, by the execution of Algorithm~\ref{alg:param_selection}, for each $i\in\range{\gamma}$, \begin{align}
        \abs{\ln\frac{1-\hat{p}_i}{\hat{p}_i}-\sum_{j\in\range{i}}\ln \frac{1-\hat{p}_j}{\hat{p}_j} - \frac{\epsilon^\dagger}{\Delta}}\leq \frac{\xi}{\Delta}.
    \end{align} Therefore, by taking the maximum over $i\in\range{\gamma}$, \begin{equation}
        \frac{\epsilon^\dagger - \xi}{\Delta} \leq \frac{\hat{\epsilon}}{\Delta} \leq \frac{\epsilon^\dagger + \xi}{\Delta}.
    \end{equation} Also, by the termination condition (Inequality~\eqref{eq:selection-termination}), \begin{align}
        \delta^\dagger \geq & \prod_{i=\lceil \log_2 \Delta \rceil}^{\gamma-1}\hat{p}_i = \sum_{j=2^\gamma - 2^{\ceil{\log_2\Delta}}}^{2^\gamma - 1}\qof{j} \geq \sum_{j=2^\gamma - \Delta}^{2^\gamma - 1}\qof{j} = \hat{\delta},
    \end{align} since $\prod_{i=\lceil \log_2 \Delta \rceil}^{\gamma-1}\hat{p}_i$ corresponds to the probability that $B_{+, i} = 1$ for all $\ceil{\log_2\Delta}\leq i \leq \gamma - 1$, which corresponds to the event that $X_+ \geq 2^\gamma - 2^{\ceil{\log_2 \Delta}}$ in Algorithm~\ref{alg:vddlm_circuit}.
\end{proof}

\begin{Lem}[Tightness of Privacy Bounds] \label{lem:tightness}
    The privacy parameters $\left(\hat{\epsilon}, \hat{\delta}\right)$ established in Theorem~\ref{thm:vddlm-dp} are asymptotically tight. Specifically, there exist lower bounds $\underline{\hat{\epsilon}} \geq \frac{\epsilon^\dagger-\xi}{\epsilon^\dagger+\xi}\hat{\epsilon}$ and $\underline{\hat{\delta}} \geq \frac{1}{2}\hat{\delta}$.
\end{Lem}

\begin{proof}[Proof of Lemma~\ref{lem:tightness}]
    Without loss of generality, we consider the case $d=1$, for a pair of neighbouring databases $D\simeq D'$ such that $\mathcal{Q}\left(D'\right) - \mathcal{Q}\left(D\right) = \Delta$. The multidimensional case ($d > 1$) reduces to the single-dimensional case by considering neighbouring databases $D\simeq D'$ such that $\mathcal{Q}\left(D\right)$ and $\mathcal{Q}\left(D'\right)$ differ in only one coordinate by $\Delta$, so the same argument applies.
    
    We derive a lower bound on privacy loss by considering the specific subset of outcomes $\mathcal{Q}(D) + S$, where $S = \left[-2^\gamma + 1, 0\right]$, such that \begin{align}
        ~ & \prob{\Mof{D} \in \mathcal{Q}(D) + S}\\
        = & \sum_{\substack{i, j\in\range{2^\gamma}\\i-j\leq 0}}\qof{i}\qof{j}\\
        = & \sum_{\substack{i, j\in\range{2^\gamma - \Delta}\\i-j\leq 0}}\qof{i}\qof{j} + \sum_{\substack{i, j\in\range{2^\gamma}\\i-j\leq 0\\j\geq 2^\gamma -\Delta }}\qof{i}\qof{j}.
    \end{align}

    By Lemma~\ref{lem:q-decreasing}, the first term is lower-bounded as \begin{align}
         ~ & \sum_{\substack{i, j\in\range{2^\gamma - \Delta}\\i-j\leq 0}}\qof{i}\qof{j} \\
         \geq & \exp\left(\frac{\epsilon^\dagger - \xi}{\Delta}\cdot \Delta\right)\sum_{\substack{i, j\in\range{2^\gamma - \Delta}\\i-j\leq 0}}\qof{i}\qof{j + \Delta}\\
         = & \exp\left(\frac{\epsilon^\dagger - \xi}{\Delta}\cdot \Delta\right)\sum_{\substack{i, j\in\range{2^\gamma}\\i-j\leq -\Delta}}\qof{i}\qof{j}\\
         = & \exp\left(\frac{\epsilon^\dagger - \xi}{\Delta}\cdot \Delta\right) \sum_{s = -2^\gamma + 1}^{-\Delta}\pof{s}\\
         = & \exp\left(\frac{\epsilon^\dagger - \xi}{\Delta}\cdot \Delta\right) \prob{\Mof{D'}\in \mathcal{Q}\left(D'\right) + S}.
    \end{align}

    Meanwhile, for the second term, by the parameter-selection termination condition in Inequality~\eqref{eq:selection-termination}, we know that $\ceil{\log_2\Delta}\leq \gamma - 1$, which implies $\Delta \leq 2^{\gamma-1}$. Because the probability mass function $\q$ is monotonically decreasing, the probability mass concentrated in the first $j$ elements is bounded below by their proportional share of the uniform distribution over the entire domain. Thus, for any $j \geq 2^\gamma - \Delta$:
    \begin{equation}
        \sum_{i\in \range{j}}\qof{i} \geq \frac{j}{2^\gamma} \sum_{i\in \range{2^\gamma}}\qof{i} \geq \frac{j}{2^\gamma}.
    \end{equation}
    Since $j \geq 2^\gamma - \Delta \geq 2^\gamma - 2^{\gamma-1} = 2^{\gamma-1}$, we have $\frac{j}{2^\gamma} \geq \frac{1}{2}$. Multiplying by $\qof{j}$ and summing over all $j \geq 2^\gamma - \Delta$ yields:
    \begin{equation}
        \sum_{\substack{i, j\in\range{2^\gamma}\\i-j\leq 0\\j\geq 2^\gamma -\Delta }}\qof{i}\qof{j} \geq \sum_{j=2^\gamma-\Delta}^{2^\gamma-1}\sum_{i\in \range{j}}\qof{j}\qof{i} \geq \frac{1}{2}\sum_{j=2^\gamma-\Delta}^{2^\gamma-1}\qof{j}.
    \end{equation}

    Therefore, combining these results gives the following lower bounds on the privacy budget: \begin{align}
        \underline{\hat{\epsilon}} &:= \epsilon^\dagger - \xi = \frac{\epsilon^\dagger - \xi}{\epsilon^\dagger + \xi} \left(\epsilon^\dagger + \xi\right) \geq \frac{\epsilon^\dagger - \xi}{\epsilon^\dagger + \xi} \cdot \hat{\epsilon} \label{eq:eps-lb}\\
        \underline{\hat{\delta}} &:= \frac{1}{2}\sum_{j=2^\gamma-\Delta}^{2^\gamma-1}\qof{j} = \frac{1}{2}\hat{\delta} \label{eq:delta-lb},
    \end{align} where Inequality~\eqref{eq:eps-lb} follows from Theorem~\ref{thm:vddlm-dp-general} directly.
\end{proof}

\begin{proof}[Proof of Theorem~\ref{thm:vddlm-complexity}]
    To connect our theoretical analysis with the implemented parameter-selection scheme, we use an error tolerance factor $\xi$ that is a small fraction of $\epsilon^\dagger$. Assuming that during the computation of $\hat{p}_i$ we account for all previously accumulated errors and maintain this error tolerance $\xi$, we are guaranteed:

    \begin{equation}
        \abs{\ln \frac{\hat{p}_i}{1-\hat{p}_i} - \ln \frac{{p}_i^\dagger}{1-p^\dagger_i}} \leq 2^i \frac{\xi}{\Delta}.
    \end{equation}

    This inequality implies $-2^i \frac{\epsilon^\dagger + \xi}{\Delta }\leq \ln \frac{\hat{p}_i}{1-\hat{p}_i} \leq - 2^i \frac{\epsilon^\dagger - \xi}{\Delta }$. Since $\xi\in\bigO{\epsilon^\dagger}$, the required bit precision for each step scales as $\nu_i=\bigO{2^i\frac{\epsilon^\dagger}{\Delta} + \log \frac{\Delta}{\xi}}$.

    Moreover, the inequality above implies $\hat{p}_i \leq \exp\left(-2^i\frac{\epsilon^\dagger-\xi}{\Delta}\right)$ for every $i$. Therefore,
    \begin{equation}
        \prod_{i=\ceil{\log_2\Delta}}^{\gamma-1}\hat{p}_i \leq \exp\left(-\sum_{i=\ceil{\log_2\Delta}}^{\gamma-1}2^i\frac{\epsilon^\dagger-\xi}{\Delta}\right).
    \end{equation}
    Since Algorithm~\ref{alg:param_selection} stops once the left-hand side is at most $\delta^\dagger$, the stopping index $\gamma$ is at most the first index satisfying
    \begin{equation}
        \sum_{i=\ceil{\log_2\Delta}}^{\gamma-1}2^i\frac{\epsilon^\dagger-\xi}{\Delta} \geq \ln \frac{1}{\delta^\dagger}.
    \end{equation}
    Given that $\xi$ is at most a constant fraction of $\epsilon^\dagger$, the left-hand side is $\bigTheta{\frac{2^\gamma \epsilon^\dagger}{\Delta}}$. Consequently, $2^\gamma \in \bigO{\frac{\Delta}{\epsilon^\dagger}\log\frac{1}{\delta^\dagger}}$, and hence
    \begin{equation}
        \gamma \in \bigO{\log\frac{\Delta}{\epsilon^\dagger} + \log\log\frac{1}{\delta^\dagger}}.
    \end{equation}

    Therefore, the total number of bits per dimension, $\sum_{i\in \range{\gamma}}\nu_i \in \bigO{\frac{2^\gamma}{\Delta} \epsilon^\dagger + \gamma\log \frac{\Delta}{\xi}}$, simplifies to:

    \begin{equation}
        \bigO{\log\frac{1}{\delta^\dagger} + \log \frac{\Delta}{\xi}\left(\log \frac{\Delta}{\epsilon^\dagger} + \log\log\frac{1}{\delta^\dagger}\right)},
    \end{equation}
    which is also
    \begin{equation}
        \bigO{\log\frac{1}{\hat\delta} + \log \frac{\Delta}{\xi}\left(\log \frac{\Delta}{\hat\epsilon} + \log\log\frac{1}{\hat\delta}\right)},
    \end{equation}
    since the parameter-selection scheme guarantees $\abs{\epsilon^\dagger - \hat{\epsilon}} \leq \xi$ and $\hat{\delta} \leq \delta^\dagger$.
\end{proof}

\end{document}